\documentclass{article}

\usepackage[all]{xy}

\SelectTips{cm}{}

\usepackage{amsbsy,mathrsfs,latexsym,amssymb,mathbbol}
\usepackage{stmaryrd,url}

\usepackage{color}
\definecolor{Blue}{rgb}{0,0,0.9}

\usepackage{amsthm}
\usepackage{fullpage}
\begin{document}
\theoremstyle{plain}
\newtheorem{theorem}{Theorem}[section]
\newtheorem{proposition}[theorem]{Proposition}
\newtheorem{corollary}[theorem]{Corollary}
\newtheorem{lemma}[theorem]{Lemma}

\theoremstyle{definition}
\newtheorem{definition}[theorem]{Definition}
\newtheorem{example}[theorem]{Example}
\newtheorem{examples}[theorem]{Examples}
\newtheorem{assumption}[theorem]{Assumption}
\newtheorem{assumptions}[theorem]{Assumptions}
\newtheorem{remark}[theorem]{Remark}
\newtheorem{remarks}[theorem]{Remarks}
\newtheorem{notation}[theorem]{Notation}
\newtheorem{observation}[theorem]{Observation}
\newtheorem{construction}[theorem]{Construction}
\newtheorem{oproblem}[theorem]{Open Problem}
\newtheorem{problems}[theorem]{Open Problems}
\newenvironment{remarkNoNumber}{\vspace{1ex}\noindent
                          {\bf{Remark.}}\ \rm }
                          {\mbox{}\hfill\vspace{1ex}}
\newenvironment{odstavec}{\refstepcounter{theorem}\vspace{1ex}\noindent
                          {\bf\thetheorem.}\ \rm }
                         {\mbox{}\hfill\vspace{1ex}}

\newenvironment{REMARK}{\vspace{2ex}\noindent\hrule\vspace{2ex}\noindent\sf}
                       {\vspace{2ex}\noindent\mbox{}\hfill\hrule\vspace{2ex}}

\font\myHuge=cmr10 at 130pt
\font\myLarge=cmr10 at 60pt

\renewcommand{\phi}{\varphi}
\def\eps{\varepsilon}
\renewcommand{\o}{\cdot}

\def\op{{\mathit{op}}}
\def\co{{\mathit{co}}}

\def\id{{\mathit{id}}}
\def\Id{{\mathit{Id}}}
\def\const{{\mathit{const}}}
\def\inl{{\sf{inl}}\,}
\def\inr{{\sf{inr}}\,}
\def\inm{{\sf{inm}}\,}
\def\can{{\sf{can}}}
\def\ev{{\sf{ev}}}
\def\ins{{\sf{in}}}

\def\wh#1{\widehat{#1}}
\def\wt#1{\widetilde{#1}}
\def\ol#1{\overline{#1}}

\def\strukt#1{\langle #1\rangle}

\def\colim{\mathop{\mathrm{colim}}\limits}
\def\Colim#1#2{{#1}\mathop{*}{#2}}

\def\Lan#1#2{{\mathrm{Lan}}_{#1} #2}
\def\Ran#1#2{{\mathrm{Ran}}_{#1} #2}

\def\kan#1#2{\langle\!\langle #1,#2\rangle\!\rangle}

\def\Sum{\bullet}
\def\tensor{\otimes}
\def\cotensor{\pitchfork}

\def\yon{{\mathbb{y}}}
\def\mult{{\mathbb{m}}}

\def\kat#1{{\mathscr{#1}}}

\def\K{\kat{K}}
\def\A{\kat{A}}
\def\B{\kat{B}}
\def\C{\kat{C}}
\def\D{\kat{D}}
\def\E{\kat{E}}
\def\F{\kat{F}}
\def\K{\kat{K}}
\def\W{\kat{W}}
\def\X{\kat{X}}
\def\P{\kat{P}}

\def\Lang{{\mathcal{L}}}

\def\One{{\mathbb{1}}}
\def\Two{{\mathbb{2}}}

\def\strukt#1{\langle #1\rangle}
\def\less{\sqsubseteq}

\def\UU{{\mathbb{U}}}
\def\LL{\mathbb{L}}
\def\PP{\mathbb{P}}

\def\KK{{\mathsf{K}}}

\def\Set{{\mathsf{Set}}}
\def\Pos{{\mathsf{Pos}}}
\def\Pre{{\mathsf{Pre}}}
\def\Rel#1{{\mathsf{Rel}}(#1)}
\def\DL{{\mathsf{DL}}}

\renewcommand{\to}{\longrightarrow}

\def\refeq#1{(\ref{#1})}

\newcommand{\emphh}[1]{\textbf{#1}}

\renewcommand{\theequation}{\thesection.\arabic{equation}}

\renewcommand{\thesubsection}{\thesection.\Alph{subsection}}

\title{Relation Liftings on Preorders and Posets
\footnote{In terms of results and numbering, the material has appeared
  in our CALCO 2011 paper of the same title, but some typos were
  corrected and proofs and a small
  number of further comments were added.} 
}

\author{Marta B\'{\i}lkov\'{a}\\ {Institute of Computer Science, Academy of Sciences of the Czech
Republic, Prague} 
\and 
Alexander Kurz\thanks{Alexander Kurz acknowledges
        the support of EPSRC, EP/G041296/1.}\\Department of Computer Science, University of Leicester,
      United Kingdom
\and Daniela Petri\c{s}an\\Department of Computer Science, University of Leicester,
      United Kingdom
\and
 Ji\v{r}\'{\i} Velebil\thanks{Marta B\'{\i}lkov\'{a} and Ji\v{r}\'{\i} Velebil
         acknowledge the support of the grant No.~P202/11/1632
         of the Czech Science Foundation.}\\Faculty of Electrical Engineering, Czech Technical University
         in Prague, Czech Republic}

\maketitle
\pagestyle{plain}

\begin{abstract}
  The category $\Rel{\Set}$ of sets and relations can be described as
  a category of spans and as the Kleisli category for the powerset
  monad. A set-functor can be lifted to a functor on $\Rel{\Set}$ iff
  it preserves weak pullbacks. We show that these results extend to
  the enriched setting, if we replace sets by posets or
  preorders. Preservation of weak pullbacks becomes preservation of
  exact lax squares. As an application we present Moss's coalgebraic
  over posets.
\end{abstract}

\setcounter{tocdepth}{1}
\tableofcontents

\section{Introduction}
\label{sec:intro}

Relation lifting
\cite{barr:rel-alg,
ckw:wpb,herm-jaco:pred-lift} plays a
crucial role in coalgebraic logic, see eg
\cite{moss:cl,baltag:cmcs00,venema:coalg-aut}.

On the one hand, it is used to explain bisimulation: If
$T:\Set\to\Set$ is a functor, then the largest bisimulation on a
coalgebra $\xi:X\to TX$ is the largest fixed point of the operator
$(\xi\times\xi)^{-1}\circ\ol{T}$ on relations on $X$, where $\ol{T}$
is the lifting of $T$ to $\Rel{\Set}\to\Rel{\Set}$. (The precise
meaning of `lifting' will be given in the Extension
Theorem~\ref{th:extension}.)

On the other hand, Moss's coalgebraic logic \cite{moss:cl} is given by
adding to propositional logic a modal operator $\nabla$, the semantics
of which is given by applying $\ol{T}$ to the forcing relation
${\Vdash}\subseteq X\times{\cal L}$, where $\cal L$ is the set of
formulas: If $\alpha\in T(\cal L)$, then $x\Vdash\nabla\alpha\
\Leftrightarrow \ \xi(x)\mathrel{\ol{T}(\Vdash)}\alpha$.

In much the same way as $\Set$-coalgebras capture bisimulation,
$\Pre$-coalgebras and $\Pos$-coalgebras capture simulation
\cite{rutten-rel,worrell:cmcs00,hugh-jaco:simulation,klin:thesis,levy:simulation,bala-kurz:calco11}. This
suggests that, in analogy with the $\Set$-based case, a coalgebraic
understanding of logics for simulations should derive from the study
of $\Pos$-functors together with on the one hand their predicate
liftings and on the other hand their $\nabla$-operator. The study of
predicate liftings of $\Pos$-functors was begun in \cite{kkv:expr},
whereas here we lay the foundations for the $\nabla$-operator of a
$\Pos$-functor. In order to do this, we start with the notion of
monotone relation for the following reason. Let $(X,\le)$ and
$(X',\le')$ be the carriers of two coalgebras, with the preorders
$\le,\le'$ encoding the simulation relations on $X$ and $X'$,
respectively. Then a simulation between the two systems will be a
relation $R\subseteq X\times X'$ such that
${\ge\,;\,R\,;\,\ge'}\subseteq {R}$, that is, $R$ is a monotone
relation. Similarly, $\Vdash$ will be a monotone relation. To
summarise, the relations we are interested in are monotone, which
enables us to use techniques of enriched category theory (of which no
prior knowledge is assumed of the reader).

For the reasons outlined above, the purpose of the paper is to develop
the basic theory of relation liftings over preorders and posets. That
is, we replace the category $\Set$ of sets and functions by the
category $\Pre$ of preorders or $\Pos$ of posets, both with monotone
(i.e. order-preserving) functions. Section~2 introduces notation and
shows that (monotone) relations can be presented by spans and by
arrows in an appropriate Kleisli-category. Section~3 recalls the
notion of exact squares. Section~4 characterises the inclusion of
functions into relations $({-})_\diamond:\Pre\to\Rel{\Pre}$ by a
universal property and shows that the relation lifting $\ol{T}$ exists
iff $T$ satisfies the Beck-Chevalley-Condition (BCC), which says that
$T$ preserves exact squares. The BCC replaces the familiar condition
known from $\Rel{\Set}$, namely that $T$ preserves weak
pullbacks. Section~5 lists examples of functors (not) satisfying the
BCC and Section~6 gives the application to Moss's coalgebraic logic
over posets.

\medskip\noindent
{\bf Related work.}
The universal property of the embedding of a (regular)
category to the category of relations is stated in
Theorem~2.3 of~\cite{hermida}.
Theorem~\ref{th:universal_property} below generalizes
this in passing from
a category to a 
simple 2-category of (pre)orders.

Liftings of functors to categories of relations within
the realm of regular categories have also been studied
in~\cite{ckw:wpb}.

\section{Monotone relations}
\label{sec:relations}
In this section we summarize briefly the notion
of monotone relations on preorders and we show that
their resulting 2-category can be perceived in two
ways:
\begin{enumerate}
\item
Monotone relations are certain {\em spans\/}, called
{\em two-sided discrete fibrations\/}.
\item
Monotone relations form a {\em Kleisli category\/}
for a certain {\em KZ doctrine\/} on the category
of preorders.
\end{enumerate}

\begin{definition}
Given preorders $\A$ and $\B$, a
{\em monotone relation $R$ from $\A$ to $\B$\/},
denoted by
$$
\xymatrix{
\A
\ar[0,1]|-{\object @{/}}^-{R}
&
\B
}
$$
is a monotone map
$
R:\B^\op\times\A\to\Two
$
where by $\Two$ we denote the two-element poset on $\{ 0,1\}$
with $0\leq 1$.
\end{definition}

\begin{remark}
\label{rem:relation}
Unravelling the definition: for a binary relation $R$,
$R(b,a)=1$ means that $a$ and $b$ are related by $R$.
Monotonicity of $R$ then means that if $R(b,a)=1$ and
$b_1\leq b$ in $\B$ and $a\leq a_1$ in $\A$, then
$R(b_1,a_1)=1$.
\end{remark}

Relations compose in the obvious way.
Two relations as on the
left below
$$
\xymatrix{
\A
\ar[0,1]|-{\object @{/}}^-{R}
&
\B
}
\quad
\xymatrix{
\B
\ar[0,1]|-{\object @{/}}^-{S}
&
\C
}
\quad
\quad
\quad
\quad
\xymatrix{
\A
\ar[0,1]|-{\object @{/}}^-{S\o R}
&
\C
}
$$
compose to the relation on the right above by the
formula
\begin{equation}
\label{eq:relation_composition}
S\o R(c,a)=\bigvee_b R(b,a)\wedge S(c,b)
\end{equation}
hence the validity of $S\o R(c,a)$ is witnessed by at least one
$b$ such that both $R(b,a)$ and $S(c,b)$
hold. 

\begin{remark}
The supremum in formula~\refeq{eq:relation_composition} is, in fact,
exactly a coend 
in the sense of enriched category theory, see~\cite{kelly:book}.
\end{remark}

The above composition of relations is associative and it has monotone
relations $ \xymatrix{ \A \ar[0,1]|-{\object @{/}}^-{\A} & \A } $ as
units, where $\A(a,a')$ holds iff $a\leq a'$.
Moreover, the relations can be ordered pointwise: $R\to S$
means that $R(b,a)$ entails $S(b,a)$, for every $a$ and $b$.
Hence we have a 2-category of monotone relations
$\Rel{\Pre}$.

\begin{remark}
  Observe that one can form analogously the 2-category $\Rel{\Pos}$ of
  monotone relations on {\em posets\/}.  In all what follows
  one can work either with preorders or posets. We will focus on
  preorders in the rest of the paper, the modifications for posets 
  always being straightforward.
Observe that both $\Rel{\Pre}$
and $\Rel{\Pos}$ have the crucial property:
  The only isomorphism 2-cells are identities. 
\end{remark}

\noindent\emphh{Remark. }
The forgetful functor $V:\Pre\to\Set$ extends to a faithful functor
$\Rel{V}:\Rel{\Pre}\to\Rel{\Set}$ where $\Rel{\Set}$ is the usual
category of sets and relations.

\subsection{The functor $({-})_\diamond:\Pre\to\Rel{\Pre}$}
We describe now the functor
$
({-})_\diamond:\Pre\to\Rel{\Pre}
$
and show its main properties. The case of posets
is completely analogous.
For a monotone map $f:\A\to\B$ define two relations
$$
\xymatrix{
\A
\ar[0,1]|-{\object @{/}}^-{f_\diamond}
&
\B
&
&
\B
\ar[0,1]|-{\object @{/}}^-{f^\diamond}
&
\A
}
$$
by the formulas $f_\diamond(b,a)=\B(b,fa)$
and $f^\diamond(a,b)=\B(fa,b)$.

\begin{lemma}
\label{lem:diamond_adjunction} For every $f:\A\to\B$ in $\Pre$ 
there is an adjunction in $\Rel{\Pre}$
 $$ \xymatrix{ f_\diamond\dashv f^\diamond
: \B \ar[0,1]|-{\object @{/}} & \A }.$$
\end{lemma}

\begin{proof}
This is easy: observe that if $\A(a,a')=1$, then
$$
f^\diamond\cdot f_\diamond (a,a')=
\bigvee_b f^\diamond(a',b)\wedge f_\diamond(b,a)=
\bigvee_b \B(fa',b) \wedge \B(b,fa)=
\B(fa,fa')=1
$$
since $f$ is a monotone map. Hence
$\eta^f:\A\to f^\diamond\cdot f_\diamond$ holds.

For the comparison $f_\diamond\cdot f^\diamond\to\B$, suppose that
$$
f_\diamond\cdot f^\diamond (b,b') =
\bigvee_a f_\diamond(b,a)\wedge f^\diamond(a,b') =
\bigvee_a \B(b,fa)\wedge \B(fa,b') =
1
$$
and use the transitivity of the order on $\B$ to conclude
that $\B(b,b')=1$.

It is now easy to show that the triangle equalities
$$
\vcenter{
\xymatrix{
f_\diamond
\ar[0,1]^-{f_\diamond\eta^f}
\ar @{=} [1,1]
&
f_\diamond\cdot f^\diamond\cdot f_\diamond
\ar[1,0]^{\eps^f f_\diamond}
\\
&
f_\diamond
}
}
\quad
\mbox{and}
\quad
\vcenter{
\xymatrix{
f^\diamond
\ar[0,1]^-{\eta^f f^\diamond}
\ar @{=} [1,1]
&
f^\diamond\cdot f_\diamond\cdot f^\diamond
\ar[1,0]^{f^\diamond \eps^f}
\\
&
f^\diamond
}
}
$$
hold and they witness the adjunction $f_\diamond\dashv f^\diamond$.
\end{proof}

\begin{remark}\label{rmk:diamond_adjunction}
Left adjoint morphisms in $\Rel{\Pre}$ 
can be
characterized as exactly those
of the form $f_\diamond$ for some monotone map $f$. 
Therefore, 
if
$
\xymatrix{
L\dashv R :
\B
\ar[0,1]|-{\object @{/}}
&
\A
}
$
in $\Rel{\Pre}$,
then there exists a monotone map
$f:\A\to\B$ such that $f_\diamond=L$ and $f^\diamond=R$.
Moreover, $f$ is uniquely determined by $L,R$ iff $\B$ is a poset.

To prove the claim, 
denote by $\eta:\A\to R\cdot L$ the unit and by
$\eps:L\cdot R\to\B$ the counit of $L\dashv R$.
First we prove that for every $a$ there is 
a 
$b_0$
such that
$$
R(a,b_0)\wedge L(b_0,a) = 1
$$
{and that $b_0$ is unique up to isomorphism:}
\begin{enumerate}
\item
Due to $\eta$ there is at least one $b$ such that
$$
R(a,b)\wedge L(b,a) = 1
$$
holds: since $\A(a,a)=1$, it is the case that
$R\cdot L(a,a)=1$.
\item
Suppose that
$$
R(a,b_1)\wedge L(b_1,a) = 1
\quad
\mbox{and}
\quad
R(a,b_2)\wedge L(b_2,a) = 1
$$
hold. Therefore the equalities
$$
R(a,b_1)\wedge L(b_2,a) = 1
\quad
\mbox{and}
\quad
R(a,b_2)\wedge L(b_1,a) = 1
$$
hold as well.
Then, due to $\eps$, we have that $\B(b_1,b_2)=1$
and $\B(b_2,b_1)=1$. 
{
In other words, we  have $b_1\le b_2$ and $b_2\le b_1$, that is, $b_1\cong b_2$ and, if $\B$ is a poset then, 
}
using antisymmetry, we conclude that $b_1=b_2$.
\end{enumerate}
Define $fa=b_0${, which determines $f$ uniquely iff $\B$ is a poset}. That the assignment $a\mapsto fa$
is monotone, follows from the existence of $\eta$.
{Finally, we need to prove $L=f_\diamond$, that is, $L(b,a)=\B(b,fa)$ for all $b,a$. We know  $L(fa,a)$ and $R(a,fa)$  by definition of $f$. 
Suppose $\B(b,fa)$, then $L(b,a)$ follows by monotonicity of $L$.
Conversely, suppose $L(b,a)$. Using $\eps:L\cdot R\to\B$, we have $\bigvee_a L(b,a)\wedge R(a,b') \le \B(b,b')$ and
choosing $b'=fa$, we get
$1= L(b,a)\wedge R(a,fa) \le \B(b,fa)$.
}
\qed
\end{remark}

Observe that if $f\to g$, then $f_\diamond\to g_\diamond$
holds. For if $\B(b,fa)=1$ then $\B(b,ga)=1$ holds
by transitivity, since $fa\leq ga$
holds.
Moreover, taking the lower diamond clearly maps an identity
monotone map $\id_\A:\A\to\A$ to the identity monotone
relation
$
\xymatrix{
\A
\ar[0,2]|-{\object @{/}}^-{\A=(\id_\A)_\diamond}
&
&
\A
}
$.
Further, taking the lower diamond preserves
composition:
$$
(g\cdot f)_\diamond(c,a)=
\C(c,gfa)=
\bigvee_b \C(c,gb) \wedge \B(b,fa)=
g_\diamond\cdot f_\diamond (c,a)
$$
Hence we have a functor $ ({-})_\diamond: \Pre\to\Rel{\Pre} $ enriched
in preorders.
Moreover, $({-})_\diamond$ is {\em locally fully
  faithful\/}, i.e., $f_\diamond\to g_\diamond$ holds iff $f\to g$
holds.

\subsection{$\Rel{\Pre}$ as a Kleisli category}

The 2-functor $({-})_\diamond:\Pre\to\Rel{\Pre}$ is a
{\em proarrow equipment with power objects\/} in the
sense of Section~2.5~\cite{marmolejo+rosebrugh+wood}.
This means that $({-})_\diamond$ has a right adjoint
$({-})^\dagger$ such that the resulting 2-monad
on $\Pre$ is a KZ doctrine and $\Rel{\Pre}$
is (up to equivalence) the corresponding Kleisli
2-category.
All of the following results are proved in the
paper~\cite{marmolejo+rosebrugh+wood}, we summarize it here
for further reference.

The 2-functor $({-})^\dagger$ works as follows:
\begin{enumerate}
\item
On objects, $\A^\dagger = [\A^\op,\Two]$, the lowersets
on $\A$, ordered by inclusion.
\item
For a relation $R$ from $\A$ to $\B$, the functor
$R^\dagger:[\A^\op,\Two]\to [\B^\op,\Two]$ is defined
as the left Kan extension of $a\mapsto R({-},a)$
along the Yoneda embedding $\yon_\A:\A\to [\A^\op,\Two]$.
This can be expressed by the formula:
$$
R^\dagger (W) =
b\mapsto \bigvee_a Wa \wedge R(b,a)
$$
i.e., $b$ is in the lowerset $R^\dagger (W)$ iff there
exists $a$ in $W$ such that $R(b,a)$ holds.
\end{enumerate}
It is easy to prove that $({-})^\dagger$ is a 2-functor
and that $({-})^\dagger\dashv ({-})_\diamond$ is a 2-adjunction
of a KZ type. The latter means that if we denote by
\begin{equation}
\label{eq:KZ}
(\LL,\yon,\mult)
\end{equation}
the resulting 2-monad on $\Pre$, then we obtain the string of
adjunctions
$
\LL(\yon_\A)
\dashv
\mult_\A
\dashv
\yon_{\LL\A}
$,
see~\cite{marmolejo1}, \cite{marmolejo2}, for more details.

The unit of the above KZ doctrine is the Yoneda embedding
$\yon_\A:\A\to [\A^\op,\Two]$ and the multiplication
$\mult_A:[[\A^\op,\Two]^\op,\Two]\to [\A^\op,\Two]$
is 
the left Kan extension of identity
on $[\A^\op,\Two]$ along 
$\yon_{[\A^\op,\Two]}$.
In more detail:
$$
\mult_\A (\W) =
a\mapsto \bigvee_W \W(W)\wedge W(a)
$$
where $\W$ is in $[[\A^\op,\Two]^\op,\Two]$ and
$W$ is in $[\A^\op,\Two]$. Hence $a$ is in the
lowerset $\mult_\A (\W)$ iff there exists
a lowerset $W$ in $\W$ such that $a$ is in $W$.
The following result is proved in Section~2.5
of~\cite{marmolejo+rosebrugh+wood}:

\begin{proposition}
\label{prop:relations_are_kleisli}
The 2-functor $({-})_\diamond:\Pre\to\Rel{\Pre}$
exhibits $\Rel{\Pre}$ as a Kleisli category
for the KZ doctrine $(\LL,\yon,\mult)$.
\end{proposition}

\subsection{Relations as spans}
Monotone relations are going to be exactly certain spans, called
{\em two-sided discrete fibrations\/}
\cite{street:fibrations}. 

\begin{definition}
\label{def:span}
A {\em span\/}
$(d_0,\E,d_1):\B\to \A$ from $\B$ to $\A$
is 
a diagram
$$
\xymatrix@R=10pt@C=15pt{
&
\E
\ar[1,-1]_{d_0}
\ar[1,1]^{d_1}
&
\\
\A
&
&
\B
}
$$
of monotone maps. The preorder
$\E$ is called the {\em vertex\/} of the span $(d_0,\E,d_1)$.
\end{definition}

\begin{remark}
Given a span $(d_0,\E,d_1):\B\to\A$, the following intuitive
notation might prove useful:
a typical element of $\E$ will be denoted by
a wiggly arrow
$$
\xymatrix{
d_0(e)
\ar @{~>} [0,1]^-{e}
&
d_1(e)
}
$$
and $d_0(e)$ will be the {\em domain\/} of $e$
and $d_1(e)$ the {\em codomain\/} of $e$.
\end{remark}

\begin{definition}
\label{def:fibrations}
A span $(d_0,\E,d_1):\B\to \A$ in $\Pre$ is a {\em two-sided discrete
  fibration\/} (we will say just {\em fibration\/} in what follows),
if the following three conditions are satisfied. For every
situation below on the left, there is a unique fill in on the
right, denoted by $(d_0)_*(e')$, respectively $(d_1)_*(e)$:
$$
\xymatrix{
a
\ar[1,0]
&
\\
a'
\ar @{~>} [0,1]_{e'}
&
b'
}
\quad\quad
\quad\quad
\quad\quad
\xymatrix{
a
\ar @{~>} [0,1]^{(d_0)_*(e')}
\ar[1,0]
&
b'
\ar @{=} [1,0]
\\
a'
\ar @{~>} [0,1]_{e'}
&
b'
}
$$

$$
\xymatrix{
a
\ar @{~>} [0,1]^{e}
&
b
\ar[1,0]
\\
&
b'
}
\quad\quad
\quad\quad
\quad\quad
\xymatrix{
a
\ar @{~>} [0,1]^{e}
\ar @{=} [1,0]
&
b
\ar[1,0]
\\
a
\ar @{~>} [0,1]_{(d_1)_*(e)}
&
b'
}
$$
Every situation on the left can be written as depicted on the right:
$$
\xymatrix{
a
\ar @{~>} [0,1]^{e}
\ar[1,0]
&
b
\ar[1,0]
\\
a'
\ar @{~>} [0,1]_{e'}
&
b'
}
\quad\quad
\quad\quad
\quad\quad
\xymatrix{
a
\ar @{~>} [0,1]^{e}
\ar @{=} [1,0]
&
b
\ar[1,0]
\\
a
\ar @{~>} [0,1]
\ar[1,0]
&
b'
\ar @{=} [1,0]
\\
a'
\ar @{~>} [0,1]_{e'}
&
b'
}
$$
\end{definition}

\noindent\emphh{Remark. }
Fibrations are jointly mono. In particular, if $\B,\A$ are discrete then $(d_0,\E,d_1):\B\to\A$ is a fibration iff it is a jointly mono.

\begin{definition}
\label{def:comma_object} A {\em comma object\/} of monotone maps
$f:\A\to\C$, $g:\B\to\C$ is a diagram
$$
\xymatrix{ f/g \ar[0,1]^-{p_1} \ar[1,0]_{p_0} \ar @{}
[1,1]|{\nearrow} & \B \ar[1,0]^{g}
\\
\A \ar[0,1]_-{f} & \C }
$$
where elements of the preorder $f/g$ are pairs $(a,b)$ with
$f(a)\leq g(b)$ in $\C$, the preorder on $f/g$ is defined
pointwise and $p_0$ and $p_1$ are the projections. 
The whole ``lax commutative square'' as above will be called a
{\em comma square\/}.
\end{definition}

\begin{example}
Every span $(p_0,f/g,p_1):\A\to\B$ arising
from a comma object of $f:\A\to\C$, $g:\B\to\C$
is a fibration.
\end{example}

A monotone relation $\xymatrix@1{\B\ar[0,1]|-{\object @{/}}^-{R}&\A}$
induces a fibration $(d_0,\E,d_1):\B\to\A$ with $\E=\{(a,b)\mid
R(a,b)=1\}$ ordered by $(a,b)\leq (a',b')$, if $a\leq a'$ and $b\leq
b'$; and $(d_0,\kat{E},d_1)$ induces the relation $R(a,b)=1\
\Leftrightarrow \ \exists e\in\E\,.\,d_0(e)=a, d_1(e)=b$.

\begin{proposition}
\label{prop:profunctors=fibrations} Fibrations in $\Pre$
correspond exactly to monotone relations. Moreover, if
$(d_0,\E,d_1):\B\to\A$ is the fibration corresponding to a
relation $\xymatrix@1{R:\B\ar[0,1]|-{\object @{/}}&\A}$, then
$R=(d_0)_\diamond\cdot(d_1)^\diamond$.
\end{proposition}

\begin{proof}
This is seen by the following
{\em Grothendieck construction}:
\begin{enumerate}
\item Given a relation $R:\A^\op\times\B\to\Two$,
      define the span
      $(d_0,\E,d_1):\B\to\A$ as follows:
      \begin{enumerate}
      \item Objects of $\E$ are pairs $(a,b)$, where $a$ and
            $b$ are objects of $\A$ and $\B$, respectively,
            with $R(a,b)=1$. A typical object is going to be denoted
            by
            $$
            \xymatrix{
            a
            \ar @{~>} [0,1]^-{(a,b)}
            &
            b
            }
            $$
      \item The preorder relation on $\E$: we put
            $(a,b)\leq (a',b')$, if $a\leq a'$,
            $b\leq b'$ in $\A$, $\B$, respectively.
            Diagrammatically:
            $$
            \xymatrix{
            a
            \ar @{~>} [0,1]^-{(a,b)}
            \ar[1,0]_{}
            &
            b
            \ar[1,0]^{}
            \\
            a'
            \ar @{~>} [0,1]_-{(a',b')}
            &
            b'
            }
            $$
            (where we write, e.g., $a\to a'$ to denote $a\leq a'$).
      \item The monotone maps $d_0:\kat{E}\to\kat{A}$ and
            $d_1:\kat{E}\to\kat{B}$ are then the
            obvious domain and codomain projections.
      \end{enumerate}
      We verify now that $(d_0,\kat{E},d_1)$ is a
      fibration.
      \begin{enumerate}
      \item Suppose
            $$
            \xymatrix{
            a
            \ar[1,0]
            &
            \\
            a'
            \ar @{~>} [0,1]_-{(a',b')}
            &
            b'
            }
            $$
            is given. We define the cartesian lift
            as follows:
            $$
            \xymatrix{
            a
            \ar @{~>} [0,1]^-{(a,b')}
            \ar[1,0]
            &
            b'
            \ar @{=} [1,0]
            \\
            a'
            \ar @{~>} [0,1]_-{(a',b')}
            &
            b'
            }
            $$
            Here we have used the fact that $R$ is monotone.
      \item Given
            $$
            \xymatrix{
            a
            \ar @{~>} [0,1]^-{(a,b)}
            &
            b
            \ar[1,0]
            \\
            &
            b'
            }
            $$
            and $g:b\to b'$, proceed analogously to the above: define
            the unique opcartesian lift
            as follows
            $$
            \xymatrix{
            a
            \ar @{~>} [0,1]^-{(a,b)}
            \ar @{=} [1,0]
            &
            b
            \ar[1,0]
            \\
            a
            \ar @{~>} [0,1]_{(a,b')}
            &
            b'
            }
            $$
      \item Suppose we are given a morphism
            $$
            \xymatrix{
            a
            \ar @{~>} [0,1]^-{(a,b)}
            \ar[1,0]
            &
            b
            \ar[1,0]
            \\
            a'
            \ar @{~>} [0,1]_-{(a',b')}
            &
            b'
            }
            $$
            in $\kat{E}$. Then it is straightforward to see that it is
            equal to the composite
            $$
            \xymatrix{
            a
            \ar @{~>} [0,1]^-{(a,b)}
            \ar @{=} [1,0]
            &
            b
            \ar[1,0]
            \\
            a
            \ar @{~>} [0,1]^-{(a,b')}
            \ar[1,0]
            &
            b'
            \ar @{=} [1,0]
            \\
            a'
            \ar @{~>} [0,1]_-{(a',b')}
            \ar @{<-} `l[uu] `[uu] [uu]
            &
            b'
            \ar @{<-} `r[uu] `[uu] [uu]
            }
            $$
      \end{enumerate}
\item Given a fibration
      $(d_0,\kat{E},d_1):\B\to\A$,
      consider the following definition
      $$
      R(a,b)=1
      \quad
      \mbox{iff}
      \quad
      \mbox{there is $e$ in $\E$ with $d_0(e)=a$ and $d_1(e)=b$}
      $$
      That the assignment $(a,b)\mapsto R(a,b)$ gives
      a monotone map
      $$
      R:\A^\op\times\B\to\Two
      $$
      is taken care of by the three conditions of Definition~\ref{def:fibrations}.
      In other words, we have obtained a relation from $\B$ to $\A$.
\end{enumerate}
\end{proof}

\noindent\emphh{Corollary. }  If
$(d_0,\E,d_1):\B\to\A$ is the fibration corresponding
$\xymatrix@1{R:\B\ar[0,1]|-{\object @{/}}&\A}$, then
$\Rel{V}R=\Rel{V}((d_0)_\diamond\cdot(d_1)^\diamond)=(Vd_0)_\diamond\cdot(Vd_1)^\diamond$.

\begin{proof}
  On the left we have that
  $\Rel{V}((d_0)_\diamond\cdot(d_1)^\diamond)(b,a)=1$ iff there is $w\in\E$
  such that $b\le d_0(w)$ and $d_1(w)\le a$. On the right we have that
  $((Vd_0)_\diamond\cdot(Vd_1)^\diamond)(b,a)=1$ iff there is $w\in\E$
  such that $b=d_0(w)$ and $d_1(w)=a$. Since $(d_0,\E,d_1)$ is a
  fibration the two conditions are equivalent.
\end{proof}

\begin{remark}
\label{rem:fibrations}
The proposition can be extended to any category enriched in $\Pre$.
The details are as follows.
A span $(d_0,\E,d_1):\B\to \A$ in $\Pre$ is a
{\em two-sided discrete fibration\/},
if the following three conditions are satisfied:
\begin{enumerate}
\item For each $m:\K\to \E$, $a,a':\K\to \A$, $b:\K\to \B$ and
      $\alpha:a'\to a$ such that triangles
      $$
      \xymatrix{
      \K
      \ar[0,1]^-{m}
      \ar[1,1]_{a}
      &
      \E
      \ar[1,0]^{d_0}
      &
      &
      \K
      \ar[0,1]^-{m}
      \ar[1,1]_{b}
      &
      \E
      \ar[1,0]^{d_1}
      \\
      &
      \A
      &
      &
      &
      \B
      }
      $$
      commute, there is a unique $\bar{m}:\K\to \E$ and a unique
      $d_0^*(\alpha):\bar{m}\to m$ such that
      $$
      \xymatrix{
      \K
      \ar[0,1]^-{\bar{m}}
      \ar[1,1]_{a'}
      &
      \E
      \ar[1,0]^{d_0}
      &
      &
      \K
      \ar[0,1]^-{\bar{m}}
      \ar[1,1]_{b}
      &
      \E
      \ar[1,0]^{d_1}
      \\
      &
      \A
      &
      &
      &
      \B
      }
      $$
      and
      $$
      \xymatrix{
      \K
      \ar @<1.3ex> [0,1]^-{\bar{m}}
      \ar @<-1.3ex> [0,1]_-{m}
      \ar @{} [0,1]|-{\downarrow d_0^*(\alpha)}
      &
      \E
      \ar[0,1]^-{d_0}
      &
      \A
      &
      =
      &
      \K
      \ar @<1.3ex> [0,1]^-{a'}
      \ar @<-1.3ex> [0,1]_-{a}
      \ar @{} [0,1]|-{\downarrow\alpha}
      &
      \A
      \\
      \K
      \ar @<1.3ex> [0,1]^-{\bar{m}}
      \ar @<-1.3ex> [0,1]_-{m}
      \ar @{} [0,1]|-{\downarrow d_0^*(\alpha)}
      &
      \E
      \ar[0,1]^-{d_1}
      &
      \B
      &
      =
      &
      \K
      \ar[0,1]^-{b}
      &
      \B
      }
      $$
      commute. The 2-cell $d_0^*(\alpha)$ is called the {\em cartesian lift\/}
      of $\alpha$.
\item For each $m:\K\to \E$, $a:\K\to \A$, $b,b':\K\to \B$ and
      $\beta:b\to b'$ such that triangles
      $$
      \xymatrix{
      \K
      \ar[0,1]^-{m}
      \ar[1,1]_{a}
      &
      \E
      \ar[1,0]^{d_0}
      &
      &
      \K
      \ar[0,1]^-{m}
      \ar[1,1]_{b}
      &
      \E
      \ar[1,0]^{d_1}
      \\
      &
      \A
      &
      &
      &
      \B
      }
      $$
      commute, there is a unique $\bar{m}:\K\to \E$ and a unique
      $d_1^*(\beta):m\Rightarrow \bar{m}$ such that
      $$
      \xymatrix{
      \K
      \ar[0,1]^-{\bar{m}}
      \ar[1,1]_{a}
      &
      \E
      \ar[1,0]^{d_0}
      &
      &
      \K
      \ar[0,1]^-{\bar{m}}
      \ar[1,1]_{b'}
      &
      \E
      \ar[1,0]^{d_1}
      \\
      &
      \A
      &
      &
      &
      \B
      }
      $$
      and
      $$
      \xymatrix{
      \K
      \ar @<1.3ex> [0,1]^-{m}
      \ar @<-1.3ex> [0,1]_-{\bar{m}}
      \ar @{} [0,1]|-{\downarrow d_1^*(\beta)}
      &
      \E
      \ar[0,1]^-{d_0}
      &
      \A
      &
      =
      &
      \K
      \ar[0,1]^-{a}
      &
      \A
      \\
      \K
      \ar @<1.3ex> [0,1]^-{m}
      \ar @<-1.3ex> [0,1]_-{\bar{m}}
      \ar @{} [0,1]|-{\downarrow d_1^*(\beta)}
      &
      \E
      \ar[0,1]^-{d_1}
      &
      \B
      &
      =
      &
      \K
      \ar @<1.3ex> [0,1]^-{b}
      \ar @<-1.3ex> [0,1]_-{b'}
      \ar @{} [0,1]|-{\downarrow\beta}
      &
      \B
      }
      $$
      commute. The 2-cell $d_1^*(\beta)$ is called the {\em opcartesian lift\/}
      of $\beta$.
\item Given any $\sigma:m\Rightarrow m':K\to E$, then the composite
      $d_0^*(d_0\sigma)\o d_1^*(d_1\sigma)$ is defined and it is equal to $\sigma$.
\end{enumerate}
The easiest way of treating fibrations abstractly is that
they are {\em algebras\/} for two (2-)monads simultaneously: they
are {\em two-sided modules\/} in a certain precise sense. See~\cite{street1}
and~\cite{street:fibrations}.
\end{remark}

\begin{example}
\label{ex:f_diamond=span}
Suppose that $f:\A\to\B$ is monotone. Recall the
relations
$
\xymatrix@C=15pt{
f_\diamond:\A
\ar[0,1]|-{\object @{/}}^-{}
&
\B
}
$
and
$
\xymatrix@C=15pt{
f^\diamond:\B
\ar[0,1]|-{\object @{/}}^-{}
&
\A
}.
$
Their corresponding fibrations are the spans
$$
\xymatrixcolsep{1pc}
\xymatrix@R=10pt@C=15pt{
&
\id_\B/f
\ar[1,-1]_{p_0}
\ar[1,1]^{p_1}
&
&
&
f/\id_\B
\ar[1,-1]_{p_0}
\ar[1,1]^{p_1}
&
\\
\B
&
&
\A
&
\A
&
&
\B
}
$$
arising from the respective comma squares.
\end{example}

\begin{example}
\label{ex:elementhood}
The relation $(\yon_\A)^\diamond$ from $\LL\A$
to $\A$ will be called the {\em elementhood\/}
relation and denoted by $\in_\A$,
since
$
(\yon_\A)^\diamond(a,A)
=
\LL\A(\yon_\A a,A)
=
A(a)
$
holds by the Yoneda Lemma.
\end{example}

\subsection{Composition of fibrations}\label{sec:comp-fib}
Suppose that we have two fibrations as on the left below.
We want to form their composite $\E\tensor\F$ {\em as a fibration\/}.
$$
\xymatrixcolsep{1pc}
\xymatrix@R=20pt{
&
\E
\ar[1,-1]_{d^\E_0}
\ar[1,1]^{d^\F_1}
&
&
&
\F
\ar[1,-1]_{d^\F_0}
\ar[1,1]^{d^\F_1}
&
\\
\C
&
&
\B
&
\B
&
&
\A
}
\quad\quad
\quad\quad\quad
\xymatrixcolsep{1pc}
\xymatrix@R=20pt
{
&
\E\tensor\F
\ar[1,-1]_{d^{\E\tensor \F}_0}
\ar[1,1]^{d^{\E\tensor \F}_1}
&
\\
\C
&
&
\A
}
$$
The idea is similar to the ordinary relations: the composite
is going to be a quotient of a pullback of spans, this time
the quotient will be taken by a map that is
surjective on objects, hence {\em absolutely dense\/}.

\begin{remark}
\label{rem:absolutely_dense}
A monotone map $e:\A\to\B$ is called
{\em absolutely dense\/} (see~\cite{absv}
and~\cite{bv}) iff
$$
\B(b,b')=
\bigvee_a \B(b,ea)\wedge\B(ea,b'),
$$
that is, $e$ is absolutely dense iff $e_\diamond\cdot e^\diamond=\id$.
Clearly, every monotone map surjective on objects is absolutely dense. The converse is true if $\B$ is a poset. If $\B$ is a
preorder, then $e$ is absolutely dense when each strongly connected
component of $\B$ contains at least one element in the image of $e$.
\end{remark}

In defining the composition of fibrations
we proceed as follows: construct the pullback
$$
\xymatrix{
\E\circ\F
\ar[0,1]^-{q_1}
\ar[1,0]_{q_0}
&
\F
\ar[1,0]^{d^\F_0}
\\
\E
\ar[0,1]_-{d^\E_1}
&
\B
}
$$
and define $\E\tensor\F$ to be the
following preorder:
\begin{enumerate}
\item
Objects are wiggly arrows of the form
$
\xymatrix@1{
c\ar @{~>}[0,1]
&
a
}
$
such that there exists $b\in \B$ with
 $(
 \xymatrix@1{
 c\ar@{~>}[0,1]
 &
 b
 },
 \xymatrix@1{
 b
 \ar @{~>}[0,1]
 &
 a}
 )\in \E\circ\F.
 $
\item
Put
$
\xymatrix@1{
c\ar @{~>}[0,1]
&
a
}
$
to be less or equal to
$
\xymatrix@1{
c'\ar @{~>}[0,1]
&
a'
}
$
iff $c\leq c'$ and $a\leq a'$.
\end{enumerate}
Define a monotone map $w:\E\circ\F\to\E\tensor\F$ in the obvious
way and observe that it is surjective on objects and, hence, absolutely dense.

We equip now $\E\tensor\F$ with the obvious
projections
$d^{\E\tensor\F}_0:\E\tensor\F\to\C$ and
$d^{\E\tensor\F}_1:\E\tensor\F\to\A$.
Then the following result is immediate.

\begin{lemma}\label{lem:tensor-fibrations}
The span $(d^{\E\tensor\F}_0,\E\tensor\F,d^{\E\tensor\F}_1):\A\to\C$
is a fibration.
\end{lemma}

To summarize, we have

  \medskip\noindent\textbf{Proposition. } Let $S,R$ be monotone
  relations with associated fibrations $\E,\F$. Then the relation
  associated with $\E\tensor\F$ is $S\cdot R$, that is, we can write
  $\E^{S\cdot R}=\E^S\tensor\E^R$.  

\section{Exact squares}
\label{sec:exact_squares}
The notion of {\em exact squares\/} replaces
the notion of weak pullbacks in the preorder
setting and exact squares will play a central r\^{o}le in our
extension theorem. Exact squares were introduced and studied
by Ren\'{e} Guitart in~\cite{guitart}.

\begin{definition}
  A lax square in $\Pre$
\begin{equation}
\label{ex:lax_square}
\vcenter{
\xymatrix{
\P
\ar[0,1]^-{p_1}
\ar[1,0]_{p_0}
\ar @{} [1,1]|{\nearrow}
&
\B
\ar[1,0]^{g}
\\
\A
\ar[0,1]_-{f}
&
\C
}
}
\end{equation}
is {\em exact\/} iff the canonical comparison in $\Rel{\Pre}$ below is
an iso (identity).
\begin{equation}
\vcenter{
\xymatrix{
\P
\ar[1,0]|-{\object @{/}}_{(p_0)_\diamond}
\ar @{} [1,1]|{\searrow}
&
\B
\ar[0,-1]|-{\object @{/}}_-{(p_1)^\diamond}
\ar[1,0]|-{\object @{/}}^{g_\diamond}
\\
\A
&
\C
\ar[0,-1]|-{\object @{/}}^-{f^\diamond}
}
}
\end{equation}
\end{definition}

\begin{remark}
In defining the canonical comparison, we use the adjunctions
$(p_1)_\diamond\dashv (p_1)^\diamond$ and $f_\diamond\dashv
f^\diamond$ guaranteed by Lemma~\ref{lem:diamond_adjunction}.

Using the formula~\refeq{eq:relation_composition}
we obtain an equivalent
criterion for exactness
namely that
\begin{equation}
\label{eq:exactness_coend}
\C(fa,gb)
=
\bigvee_w \A(a,p_0 w)\wedge\B(p_1 w,b)
\end{equation}
\end{remark}

\begin{example}
\label{ex:guitart}
We give examples of exact squares in $\Pre$. They all
come from Guitart's paper~\cite{guitart}, Example~1.14.
The proofs follow immediately from the
description~\refeq{eq:exactness_coend} above.
\begin{enumerate}
\item
The square
$$
\xymatrix{
\A
\ar[0,1]^-{f}
\ar[1,0]_{1_\A}
\ar @{} [1,1]|{\nearrow}
&
\B
\ar[1,0]^{1_\B}
\\
\A
\ar[0,1]_-{f}
&
\B
}
$$
where the comparison is identity, is always exact since
$$
\B(fa,b)
=
\bigvee_w \A(a,w)\wedge\B(fw,b)
$$
holds by the Yoneda Lemma. Such a square is called
a {\em Yoneda square\/} in~\cite{guitart}.
\item
The square
$$
\xymatrix{
\A
\ar[0,1]^-{1_\A}
\ar[1,0]_{f}
\ar @{} [1,1]|{\nearrow}
&
\A
\ar[1,0]^{f}
\\
\B
\ar[0,1]_-{1_\B}
&
\B
}
$$
where the comparison is identity, is always exact since
$$
\B(b,fa)
=
\bigvee_w \B(b,fw)\wedge\A(w,a)
$$
holds by the Yoneda Lemma. Again, squares
of this form are called
{\em Yoneda squares\/} in~\cite{guitart}.
\item
Every {\em comma square\/}
$$
\xymatrix{
f/g
\ar[0,1]^-{d_1}
\ar[1,0]_{d_0}
\ar @{} [1,1]|{\nearrow}
&
\B
\ar[1,0]^{g}
\\
\A
\ar[0,1]_-{f}
&
\C
}
$$
is exact.
\item
Every {\em op-comma square\/}
$$
\xymatrix{
\C
\ar[0,1]^-{g}
\ar[1,0]_{f}
\ar @{} [1,1]|{\nearrow}
&
\B
\ar[1,0]^{i_1}
\\
\A
\ar[0,1]_-{i_0}
&
f\triangleright g
}
$$
is exact.
\item
The square
$$
\xymatrix{
\A
\ar[0,1]^-{1_\A}
\ar[1,0]_{1_\A}
\ar @{} [1,1]|{\nearrow}
&
\A
\ar[1,0]^{f}
\\
\A
\ar[0,1]_-{f}
&
\B
}
$$
(where the comparison is identity) is exact iff
$f$ is an {\em order-embedding\/}, i.e., iff the following
holds: $fa\leq fa'$ iff $a\leq a'$.

Such $f$'s can also be called {\em fully faithful\/}.
\item
\label{item:abs_dense}
The square
$$
\xymatrix{
\A
\ar[0,1]^-{e}
\ar[1,0]_{e}
\ar @{} [1,1]|{\nearrow}
&
\B
\ar[1,0]^{1_\B}
\\
\B
\ar[0,1]_-{1_\B}
&
\B
}
$$
(where the comparison is identity) is exact iff
$e$ is {\em absolutely dense\/}, i.e., iff 
$$
\B(b,b')
=
\bigvee_a \B(b,ea)\wedge\B(ea,b').
$$
See, e.g., \cite{absv}
and~\cite{bv} for more details on absolutely dense
maps.
\item\label{item:adjunction}
The square
$$
\xymatrix{
\X
\ar[0,1]^-{f}
\ar[1,0]_{1_\X}
\ar @{} [1,1]|{\nearrow}
&
\A
\ar[1,0]^{u}
\\
\X
\ar[0,1]_-{1_\X}
&
\X
}
$$
is exact iff $f\dashv u:\A\to\X$ holds. Moreover, the
comparison in the above square is the unit of $f\dashv u$.
\item
The square
$$
\xymatrix{
\A
\ar[0,1]^-{1_\A}
\ar[1,0]_{u}
\ar @{} [1,1]|{\nearrow}
&
\A
\ar[1,0]^{1_\A}
\\
\X
\ar[0,1]_-{f}
&
\A
}
$$
is exact iff $f\dashv u:\A\to\X$ holds. Moreover, the
comparison in the above square is the counit of $f\dashv u$.
\item
The square
$$
\xymatrix{
\X'
\ar[0,1]^-{f}
\ar[1,0]_{1_{\X'}}
\ar @{} [1,1]|{\nearrow}
&
\A
\ar[1,0]^{u}
\\
\X'
\ar[0,1]_-{j}
&
\X
}
$$
is exact iff $f\dashv_j u:\A\to\X$ holds, i.e., iff
$f$ is a left adjoint of $u$ {\em relative to\/} $j$.

In general, relative adjointness means the existence of an isomorphism
$$
\X(jx',ua)
\cong
\A(fx',a)
$$
natural in $x'$ and $a$, and due to
$$
\A(fx',a)
\cong
\bigvee_{w}\X'(w,x')\wedge\A(fw,a)
$$
this means precisely the exactness of the above square.
\item
The square
$$
\xymatrix{
\A
\ar[0,1]^-{j}
\ar[1,0]_{h}
\ar @{} [1,1]|{\nearrow}
&
\B
\ar[1,0]^{l}
\\
\X
\ar[0,1]_-{1_\X}
&
\X
}
$$
is exact iff the comparison exhibits $l$ as
an {\em absolute\/} left Kan extension of $h$
along $j$. In fact, 
$$
\X(x,lb)
=
\bigvee_a \X(x,ha)\wedge\B(ja,b)
$$
asserts precisely that
\begin{enumerate}
\item
$l$ is a left Kan extension of $h$ along $j$.

For any $k:\B\to\X$ we need to prove $l\to k$ iff $h\to k\o j$.
\begin{enumerate}
\item Suppose $lb\leq kb$ for all $b$. Choose any $a$. Then
$ha\leq lja$ by the square above. Since $lja\leq kja$ by
assumption, hence $ha\leq kja$. \item Suppose $ha\leq kja$ for all
$a$. To prove $lb\leq kb$ for all $b$, it suffices to prove that
$x\leq lb$ implies $x\leq kb$, for all $x$. Suppose $x\leq lb$,
i.e., $\X(x,lb)=1$. Hence $\bigvee_a\X(x,ha)\wedge\B(ja,b)=1$.
Choose $a$ to witness $x\leq ha$ and $ja\leq b$. From our
assumption we obtain $x\leq kja$, hence $x\leq kb$.
\end{enumerate}
\item
$l$ is an absolute left Kan extension of $h$ along $j$.

We need to prove that for any $f:\X\to\X'$, $f\o l$ is a left
Kan extension of $f\o h$ along $j$. That is,
for any $k:\B\to\X'$ we need to prove $f\o l\to k$ iff $f\o h\to k\o j$.

This is proved in the same manner as above.
\end{enumerate}
Observe that item~\ref{item:adjunction} above is a special
case of absolute Kan extensions by B\'{e}nabou's Theorem:
$f\dashv u$ holds if the unit exhibits $u$ as an absolute
left Kan extension of identity along $f$.
\end{enumerate}
\end{example}

\begin{example}
\label{ex:square+adjoints}
Every square~\refeq{ex:lax_square} where $f$ and $p_1$ are
{\em left\/} adjoints, is exact
iff $p_0\o p_1^r
=
f^r\o g$, where we
denote by $f^r$ and $p_1^r$ the respective right adjoints.

This is proved as follows. Firstly, the comparison $f\o p_0\to g\o
p_1$ is equivalent to the comparison $p_0\o p_1^r\to f^r\o g$ due to
adjunctions $f\dashv f^r$ and $p_1\dashv p_1^r$.
Further, we have 
$$
\bigvee_w \A(a,p_0 w)\wedge \B(p_1 w,b) 
=
\bigvee_w \A(a,p_0
w)\wedge \P(w,p_1^r b) 
=
\A(a,p_0 p_1^r b)
$$
and
$$
\C(fa,gb)
=
\A(a,f^r gb)
$$
It follows that the square~\refeq{ex:lax_square} is exact
iff 
$$
\A(a,p_0 p_1^r b)
=
\A(a,f^r gb).
$$
By the Yoneda Lemma, this is equivalent
to 
$p_0\o p_1^r
=
f^r\o g$.
\end{example}

\begin{example}
\label{ex:dual} If the square on the left is exact, then so is the
square on the right:
$$
\xymatrix{
\P
\ar[0,1]^-{p_1}
\ar[1,0]_{p_0}
\ar @{} [1,1]|{\nearrow}
&
\B
\ar[1,0]^{g}
\\
\A
\ar[0,1]_-{f}
&
\C
}
\quad\quad
\quad\quad
\quad\quad
\xymatrix{
\P^\op
\ar[0,1]^-{p_0^\op}
\ar[1,0]_{p_1^\op}
\ar @{} [1,1]|{\nearrow}
&
\A^\op
\ar[1,0]^{f^\op}
\\
\B^\op
\ar[0,1]_-{g^\op}
&
\C^\op
}
$$
\end{example}

\noindent To prove the claim, by~\refeq{eq:exactness_coend}, we need
$$
\C^\op (g^\op b,f^\op a)
=
\bigvee_w \B^\op (b,p_1^\op w)\wedge\A^\op(p_0^\op w,a)
$$
But
$$
\C^\op (g^\op b,f^\op a)
=
\C(fa,gb)
$$
and
$$
\bigvee_w \B^\op (b,p_1^\op w)\wedge\A^\op(p_0^\op w,a)
=
\bigvee_w \A(a,p_0 w)\wedge\B(p_1 w,b)
$$
and this finishes the proof.

\begin{lemma}
\label{lem:BCC}
Suppose that $(d_0^S,\E^S,d_1^S)$ and $(d_0^R,\E^R,d_1^R)$
are two-sided discrete fibrations.
Then the pullback
$$
\xymatrix{ \E^S\circ \E^R \ar[0,1]^-{q_1} \ar[1,0]_{q_0} & \E^R
\ar[1,0]^{d^R_0}
\\
\E^S
\ar[0,1]_-{d^S_1}
&
\B
}
$$
considered as a lax commutative square
where the comparison is identity, is exact.
\end{lemma}

\begin{proof}
Suppose that $d^S_1(e)\leq d^R_0(f)$ holds. Then we have a
situation
$$
\xymatrix{
c
\ar @{~>} [0,1]^-{e}
&
b\leq b'
\ar @{~>} [0,1]^-{f}
&
a
}
$$
and there exists $w$ in $\E^S\circ\E^R$ of the form
$$
\xymatrix{
c
\ar @{~>} [0,1]^-{e'}
&
b'
\ar @{~>} [0,1]^-{f}
&
a
}
$$
that clearly satisfies
$e\leq p_0(e',f)$ and $p_1(e',f)\leq f$.
\end{proof}

Given monotone relations $\xymatrix@1{\A\ar[0,1]|-{\object @{/}}^-{R}&\B}$ and $\xymatrix@1{\B\ar[0,1]|-{\object @{/}}^-{S}&\C}$, the two-sided fibration corresponding to the composition $S\cdot R$ is the composition of the fibrations corresponding to $S$ and $R$ as described in Section~\ref{sec:comp-fib}. The properties described in the next Corollary are essential for the proof of Theorem~\ref{th:universal_property}.

\begin{corollary}
\label{cor:composition}
Form, for a pair $R$, $S$, of monotone relations
the following commutative diagram
$$
\xymatrix@R=18pt{
&
&
\E^{S\o R}
\ar `l[llddd] [llddd]_{d^{S\o R}_0}
\ar `r[rrddd] [rrddd]^{d^{S\o R}_1}
&
&
\\
&
&
\E^S\circ\E^R
\ar[1,-1]_{q_0}
\ar[1,1]^{q_1}
\ar[-1,0]_{w}
&
&
\\
&
\E^S
\ar[1,-1]_{d^S_0}
\ar[1,1]^{d^S_1}
\ar @{} [0,2]|-{\to}
&
&
\E^R
\ar[1,-1]_{d^R_0}
\ar[1,1]^{d^R_1}
&
\\
\C
&
&
\B
&
&
\A
}
$$
where the lax commutative square in the middle is
a pullback square (hence the comparison is the identity),
and $w$ is a map, surjective on objects, coming from
composing $\E^S$ and $\E^R$ as fibrations.
Then the square is exact and $w$ is an absolutely dense monotone
map. 
\end{corollary}

In the extension theorem we will demand that a certain functor $T:\Pre\to\Pre$
preserves exact squares, whereas the proof of the theorem actually only
needs the at first sight weaker requirement that $T$ preserves strict
exact squares and preserves the exactness of comma squares of the
form $1_\A/1_\A$ (the former being needed for preservation of
composition and the latter for preservation of identities). It
therefore seems of interest to present the following result.

\medskip\noindent\emphh{Proposition. }
For a locally monotone $T:\Pre\to\Pre$, or $T:\Pos\to\Pos$, the following 
are equivalent:
\begin{enumerate}
\item 
$T$ preserves {\em lax\/} exact squares.

\item\label{item:crucial}
$T$ preserves strict exact squares and exactness
of comma squares of the form $1_\A/1_\A$, 
for all $\A$.
\item
$T$ preserves strict exact squares and exactness
of comma squares of the form $f/1_\B$, $1_\A/f$,
for all $f:\A\to\B$.
\item
$T$ preserves strict exact squares and 
exactness of comma squares.
\end{enumerate}

\begin{proof}
(1) implies (2): clear.

\medskip\noindent
(2) implies (3): Suppose $f:\A\to\B$ is a monotone
map. We prove that $T$ preserves exactness
of the comma square
$$
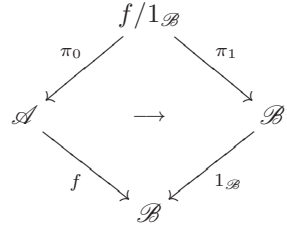
\xymatrix{
&
f/1_\B
\ar[1,-1]_{\pi_0}
\ar[1,1]^{\pi_1}
&
\\
\A
\ar[1,1]_{f}
\ar@{}[0,2]|-{\to}
&
&
\B
\ar[1,-1]^{1_\B}
\\
&
\B
&
}
$$
That $T$ preserves exactness of comma squares
of the form $1_\A/f$ is proved analogously. 

Define $e:\P\to f/1_\B$
by the universal property in
\begin{equation}
\label{eq:e0}
\vcenter{
\xymatrix{
&
\P
\ar[1,0]^{e}
&
\\
&
f/1_\B
\ar[1,-1]_{\pi_0}
\ar[1,1]^{\pi_1}
&
\\
\A
\ar[1,1]_{f}
\ar@{}[0,2]|-{\to}
&
&
\B
\ar[1,-1]^{1_\B}
\\
&
\B
&
}
}
\quad
=
\quad
\vcenter{
\xymatrix{
&
&
\P
\ar[1,-1]_{s'_1}
\ar[1,1]^{s'_0}
&
&
\\
&
\A
\ar[1,-1]_{1_\A}
\ar[1,1]^{f}
\ar@{}[0,2]^-{(i)}
&
&
1_\B/1_\B
\ar[1,-1]_{s_1}
\ar[1,1]^{p'_1}
&
\\
\A
\ar[1,1]_{f}
\ar@{}[0,2]^-{(ii)}
&
&
\B
\ar[1,-1]_{1_\B}
\ar[1,1]^{1_\B}
\ar@{}[0,2]^-{(iii)}
\ar@{}[0,2]|-{\to}
&
&
\B
\ar[1,-1]^{1_\B}
\\
&
\B
\ar[1,1]_{1_\B}
\ar@{}[0,2]^-{(iv)}
&
&
\B
\ar[1,-1]^{1_\B}
&
\\
&
&
\B
&
&
}
}
\end{equation}
where $(i)$, $(ii)$, $(iv)$ are pullbacks, and $(iii)$
is a comma square.

Clearly, $e:\P\to f/1_\B$ maps $(a,b',b)$ in $\P$ to $(a,b)$
in $f/1_\B$ and $e$ is a monotone surjection.

The image under $T$ of the diagram on the right of~\refeq{eq:e0}
is exact by assumptions. Hence the image under $T$ of the
diagram on the left of~\refeq{eq:e0} is exact. Since $e$
is a surjection, $e_\diamond\o e^\diamond=1_{f/1_{\B}}$.
Hence $(Te)_\diamond\o (Te)^\diamond=1_{T(f/g)}$ holds
since $T$ preserves surjections (express surjectivity 
as a strict exact square). Thus 
\begin{eqnarray*}
(T\pi_0)_\diamond\o (T\pi_1)^\diamond
&=&
(T\pi_0)_\diamond\o (Te)_\diamond\o (Te)^\diamond \o (T\pi_1)^\diamond
\\
&=&
(Tf)^\diamond \o (Tg)_\diamond
\end{eqnarray*}
proving exactness of
$$
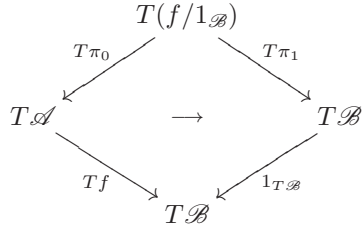
\xymatrix{
&
T(f/1_\B)
\ar[1,-1]_{T\pi_0}
\ar[1,1]^{T\pi_1}
&
\\
T\A
\ar[1,1]_{Tf}
\ar@{}[0,2]|-{\to}
&
&
T\B
\ar[1,-1]^{1_{T\B}}
\\
&
T\B
&
}
$$

\medskip\noindent
(3) implies (4): Suppose
$$
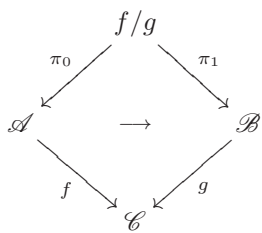
\xymatrix{
&
f/g
\ar[1,-1]_{\pi_0}
\ar[1,1]^{\pi_1}
&
\\
\A
\ar[1,1]_{f}
\ar@{}[0,2]|-{\to}
&
&
\B
\ar[1,-1]^{g}
\\
&
\C
&
}
$$
is a comma square and define $e:\P\to f/g$
by the universal property in
\begin{equation}
\label{eq:e1}
\vcenter{
\xymatrix{
&
\P
\ar[1,0]^{e}
&
\\
&
f/g
\ar[1,-1]_{\pi_0}
\ar[1,1]^{\pi_1}
&
\\
\A
\ar[1,1]_{f}
\ar@{}[0,2]|-{\to}
&
&
\B
\ar[1,-1]^{g}
\\
&
\C
&
}
}
\quad
=
\quad
\vcenter{
\xymatrix{
&
&
\P
\ar[1,-1]_{s'_1}
\ar[1,1]^{s'_0}
&
&
\\
&
f/1_\C
\ar[1,-1]_{p'_0}
\ar[1,1]^{s_0}
\ar@{}[0,2]^-{(i)}
&
&
1_\C/g
\ar[1,-1]_{s_1}
\ar[1,1]^{p'_1}
&
\\
\A
\ar[1,1]_{f}
\ar@{}[0,2]^-{(ii)}
\ar@{}[0,2]|-{\to}
&
&
\C
\ar[1,-1]_{1_\C}
\ar[1,1]^{1_\C}
\ar@{}[0,2]^-{(iii)}
\ar@{}[0,2]|-{\to}
&
&
\B
\ar[1,-1]^{g}
\\
&
\C
\ar[1,1]_{1_\C}
\ar@{}[0,2]^-{(iv)}
&
&
\C
\ar[1,-1]^{1_\C}
&
\\
&
&
\C
&
&
}
}
\end{equation}
where $(i)$ and $(iv)$ are pullbacks, $(ii)$ and $(iii)$
are comma squares.

Clearly, $e:\P\to f/g$ maps $(a,c,b)$ in $\P$ to $(a,b)$
in $f/g$ and $e$ is a monotone surjection.

The image under $T$ of the diagram on the right of~\refeq{eq:e1}
is exact by assumptions. Hence the image under $T$ of the
diagram on the left of~\refeq{eq:e1} is exact. Since $e$
is a surjection, $e_\diamond\o e^\diamond=1_{f/g}$.
Hence $(Te)_\diamond\o (Te)^\diamond=1_{T(f/g)}$ holds
since $T$ preserves surjections (express surjectivity 
as a strict exact square). Thus 
\begin{eqnarray*}
(T\pi_0)_\diamond\o (T\pi_1)^\diamond
&=&
(T\pi_0)_\diamond\o (Te)_\diamond\o (Te)^\diamond \o (T\pi_1)^\diamond
\\
&=&
(Tf)^\diamond \o (Tg)_\diamond
\end{eqnarray*}
proving exactness of  
$$
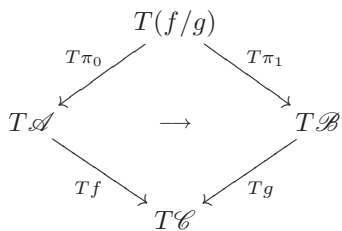
\xymatrix{
&
T(f/g)
\ar[1,-1]_{T\pi_0}
\ar[1,1]^{T\pi_1}
&
\\
T\A
\ar[1,1]_{Tf}
\ar@{}[0,2]|-{\to}
&
&
T\B
\ar[1,-1]^{Tg}
\\
&
T\C
&
}
$$

\medskip\noindent
(4) implies (1): 
Suppose that the lax square
$$
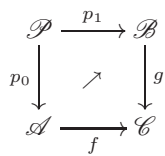
\xymatrix{
\P
\ar[0,1]^-{p_1}
\ar[1,0]_{p_0}
&
\B
\ar[1,0]^-{g}
\\
\A
\ar[0,1]_-{f}
\ar@{}[-1,1]|{\nearrow}
&
\C
}
$$
is exact. 

Observe that there is an equality
\begin{equation}
\label{eq:e2}
\vcenter{
\xymatrix{
&
\kat{S}
\ar[1,0]^{e}
&
\\
&
f/g
\ar[1,-1]_{\pi_0}
\ar[1,1]^{\pi_1}
&
\\
\A
\ar[1,1]_{f}
\ar@{}[0,2]|-{\to}
&
&
\B
\ar[1,-1]^{g}
\\
&
\C
&
}
}
\quad
=
\quad
\vcenter{
\xymatrix{
&
&
\kat{S}
\ar[1,-1]_{s'_1}
\ar[1,1]^{s'_0}
&
&
\\
&
1_\A/p_0
\ar[1,-1]_{p'_0}
\ar[1,1]^{s_0}
\ar@{}[0,2]^-{(i)}
&
&
p_1/1_\B
\ar[1,-1]_{s_1}
\ar[1,1]^{p'_1}
&
\\
\A
\ar[1,1]_{1_\A}
\ar@{}[0,2]^-{(ii)}
\ar@{}[0,2]|-{\to}
&
&
\P
\ar[1,-1]_{p_0}
\ar[1,1]^{p_1}
\ar@{}[0,2]^-{(iii)}
\ar@{}[0,2]|-{\to}
&
&
\B
\ar[1,-1]^{1_\B}
\\
&
\A
\ar[1,1]_{f}
\ar@{}[0,2]^-{(iv)}
\ar@{}[0,2]|-{\to}
&
&
\B
\ar[1,-1]^{g}
&
\\
&
&
\C
&
&
}
}
\end{equation}
where the diagrams on the right are: $(i)$ is a pullback,
$(ii)$ and $(iii)$ are comma objects, and $(iv)$ is the
original lax exact square. On the left, the morphism 
$e:\kat{S}\to f/g$ is induced by the universal property 
of comma squares. Observe that $e$ is a monotone 
surjection: $e$ maps $(a,w,b)$ in $\kat{S}$ to $(a,b)$
in $f/g$, and for $(a,b)$ in $f/g$ there is $(a,w,b)$
in $\kat{S}$ by exactness.

Therefore, the diagram
$$
\vcenter{
\xymatrix{
&
\kat{S}
\ar[1,0]^{e}
&
\\
&
f/g
\ar[1,-1]_{\pi_0}
\ar[1,1]^{\pi_1}
&
\\
\A
\ar[1,1]_{f}
\ar@{}[0,2]|-{\to}
&
&
\B
\ar[1,-1]^{g}
\\
&
\C
&
}
}
=
\vcenter{
\xymatrix{
&
\kat{S}
\ar[1,-1]_{\pi_0\o e}
\ar[1,1]^{\pi_1\o e}
&
\\
\A
\ar[1,1]_{f}
\ar@{}[0,2]|-{\to}
&
&
\B
\ar[1,-1]^{g}
\\
&
\C
&
}
}
$$ 
is exact, i.e., the equality
$$
(\pi_0)_\diamond \o e_\diamond \o e^\diamond \o (\pi_1)^\diamond
=
f^\diamond \o g_\diamond
$$
holds. This follows from $ e_\diamond \o e^\diamond = 1_{f/g}$,
since $e$ is surjective and from the fact that comma squares
are exact.

By assumption, in the diagram
$$
\xymatrix{
&
&
T(f/g)
\ar[1,-1]_{Ts'_1}
\ar[1,1]^{Ts'_0}
&
&
\\
&
T(1_\A/p_0)
\ar[1,-1]_{Tp'_0}
\ar[1,1]^{Ts_0}
\ar@{}[0,2]^-{T(i)}
&
&
T(p_1/1_\B)
\ar[1,-1]_{Ts_1}
\ar[1,1]^{Tp'_1}
&
\\
T\A
\ar[1,1]_{1_{T\A}}
\ar@{}[0,2]^-{T(ii)}
\ar@{}[0,2]|-{\to}
&
&
T\P
\ar[1,-1]_{Tp_0}
\ar[1,1]^{Tp_1}
\ar@{}[0,2]^-{T(iii)}
\ar@{}[0,2]|-{\to}
&
&
T\B
\ar[1,-1]^{1_{T\B}}
\\
&
T\A
\ar[1,1]_{Tf}
\ar@{}[0,2]^-{T(iv)}
\ar@{}[0,2]|-{\to}
&
&
T\B
\ar[1,-1]^{Tg}
&
\\
&
&
T\C
&
&
}
$$
the square $T(i)$ is strict exact, and $T(ii)$, 
$T(iii)$ are lax exact squares.
Also, the whole diagram is exact, being the image of the diagram
$$
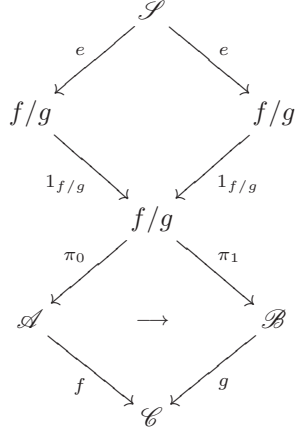
\xymatrix{
&
\kat{S}
\ar[1,-1]_{e}
\ar[1,1]^{e}
&
\\
f/g
\ar[1,1]_{1_{f/g}}
&
&
f/g
\ar[1,-1]^{1_{f/g}}
\\
&
f/g
\ar[1,-1]_{\pi_0}
\ar[1,1]^{\pi_1}
&
\\
\A
\ar[1,1]_{f}
\ar@{}[0,2]|-{\to}
&
&
\B
\ar[1,-1]^{g}
\\
&
\C
&
}
$$
under $T$ (use assumptions: the upper square is strict exact, 
and the lower square is a comma object).

We prove that $T(iv)$ is exact.
Indeed:
\begin{eqnarray*}
(Tf)^\diamond\o (Tg)_\diamond
&=&
(1_{T\A})^\diamond \o (Tf)^\diamond\o (Tg)_\diamond\o (1_{T\B})_\diamond
\\
&=&
(Tp'_0)_\diamond \o (Ts'_1)_\diamond \o (Ts'_0)^\diamond \o (Tp'_1)^\diamond
\\
&=&
(Tp'_0)_\diamond \o (Ts_0)^\diamond \o (Ts_1)_\diamond \o (Tp'_1)^\diamond
\\
&=&
(1_{T\A})^\diamond \o (Tp_0)_\diamond \o (Tp_1)^\diamond \o (1_{T\B})_\diamond
\\
&=&
(Tp_0)_\diamond \o (Tp_1)^\diamond
\end{eqnarray*}
\end{proof}

\section{The universal property of $({-})_\diamond:\Pre\to\Rel{\Pre}$}
\label{sec:universal_property}

We prove now that the 2-functor $({-})_\diamond:\Pre\to\Rel{\Pre}$
has an analogous universal property to the case of sets. From
that, the result on a unique lifting of $T$ to $\ol{T}$ will
immediately follow, see Theorem~\ref{th:extension} below.

\begin{theorem}
\label{th:universal_property}
The 2-functor $({-})_\diamond:\Pre\to\Rel{\Pre}$
has the following three properties:
\begin{enumerate}
\item
Every $f_\diamond$ is a left adjoint.
\item
For every exact square~\refeq{ex:lax_square}
the equality
$f^\diamond\o g_\diamond = (p_0)_\diamond\o (p_1)^\diamond$
holds.
\item
For every absolutely dense monotone map $e$,
the relation $e_\diamond$ is a split epimorphism
with the splitting given by $e^\diamond$.
\end{enumerate}

Moreover, the functor $({-})_\diamond$
is universal w.r.t. these three properties
in the following sense: if $\KK$ is any 2-category
where the isomorphism 2-cells are identities, to give
a 2-functor $H:\Rel{\Pre}\to\KK$ is the same thing
as to give a 2-functor $F:\Pre\to\KK$ with the following three
properties:
\begin{enumerate}
\item
Every $Ff$ has a right adjoint, denoted by $(Ff)^r$.
\item
For every exact square~\refeq{ex:lax_square}
the equality
$Ff^r\o Fg = Fp_0\o (Fp_1)^r$
holds.
\item
For every absolutely dense monotone map $e$,
$Fe$ is a split epimorphism, with the splitting
given by $(Fe)^r$.
\end{enumerate}
\end{theorem}

\begin{proof}
It is trivial to see that $({-})_\diamond$
has the above three properties.

Given a 2-functor $H:\Rel{\Pre}\to\KK$,
define $F$ to be the composite $H\o ({-})_\diamond$.
Such $F$ clearly has the above three properties,
since 2-functors preserve adjunctions.

Conversely, given $F:\Pre\to\KK$, define $H\A=F\A$
on objects, and
on a relation $R=(d_0^R)_\diamond\o (d_1^R)^\diamond$
define $H(R)=Fd_0^R\o (Fd_1^R)^r$, where
$(Fd_1^R)^r$ is the right adjoint of $Fd_1^R$ in $\KK$.

It is easy to verify that $H$ so defined preserves
identities: the identity relation $\id_\A$ on $\A$
is represented as a fibration
$$
\xymatrixcolsep{1.5pc}
\xymatrix{
&
1_\A/1_\A
\ar[1,-1]_{p_0}
\ar[1,1]^{p_1}
&
\\
\A
&
&
\A
}
$$
coming from the exact comma square
\begin{equation}\label{eq:exact-comma-square}
\xymatrix{
1_\A/1_\A
\ar[0,1]^-{p_1}
\ar[1,0]_{p_0}
\ar @{} [1,1]|{\nearrow}
&
\A
\ar[1,0]^{1_\A}
\\
\A
\ar[0,1]_-{1_\A}
&
\A
}
\end{equation}
Hence $H(\id_\A)=Fp_0\o (Fp_1)^r= F(1_\A)=1_{F\A}=1_{H\A}$
holds by our assumptions on $F$.

For preservation of composition use
Corollary~\ref{cor:composition}: first
$$
H(S)\o H(R)
=
Fd^S_0\o (Fd^S_1)^r
\o
Fd^R_0\o (Fd^R_1)^r
$$
by definition. Further, by exactness
of the pullback from Corollary~\ref{cor:composition}
and our assumption on $F$, we have
$$
Fd^S_0\o (Fd^S_1)^r
\o
Fd^R_0\o (Fd^R_1)^r
=
Fd^S_0\o Fq_0
\o
(Fq_1)^r\o (Fd^R_1)^r
$$
and, finally, since $Fw$ is split epi
by Corollary~\ref{cor:composition} and our assumption
on $F$, we obtain
$$
Fd^S_0\o Fq_0\o Fw
\o
(Fw)^r\o (Fq_1)^r\o (Fd^R_1)^r
=
Fd^{R\o S}_0\o (Fd^{R\o S}_1)^r
=
H(R\o S)
$$
and the proof is complete.
 \end{proof}

\noindent\emphh{Remark. } 
There is an analogous theorem with ``$\Pos$'' replacing ``$\Pre$'' and
``surjective'' replacing ``absolutely dense''.

\section{The extension theorem}
\label{sec:extension}

\begin{definition}
\label{def:BCC}
We say that a locally monotone functor
$T:\Pre\to\Pre$ satisfies the {\em Beck-Chevalley Condition\/} (BCC)
if it preserves exact squares.
\end{definition}

\begin{remark}\label{rmk:BCC}
  A functor satisfying the BCC 
  has to preserve order-embeddings, absolutely dense monotone maps and
  absolute left Kan extensions. This follows from Example~\ref{ex:guitart}.
Examples of functors (not) satisfying the
BCC 
can be found in Section~\ref{sec:examples}.
\end{remark}

\begin{theorem}
\label{th:extension}
For a 2-functor $T:\Pre\to\Pre$ the following
are equivalent:
\begin{enumerate}
\item
There is a 2-functor $\ol{T}:\Rel{\Pre}\to\Rel{\Pre}$
such that 
\begin{equation}
\label{eq:extension_square}
\vcenter{
\xymatrix{
\Rel{\Pre}
\ar[0,2]^-{\ol{T}}
&
&
\Rel{\Pre}
\\
\Pre
\ar[0,2]_-{T}
\ar[-1,0]^{({-})_\diamond}
&
&
\Pre
\ar[-1,0]_{({-})_\diamond}
}
}
\end{equation}
\item
The functor $T$ satisfies the BCC. 
\item
There is a distributive law
$
\
T\o\LL\to\LL\o T
\:$
of $T$ over the KZ doctrine $(\LL,\yon,\mult)$
described in~\refeq{eq:KZ} above.
\end{enumerate}
\end{theorem}
\begin{proof}
  The equivalence of 1.~and 3.~follows from general facts about
  distributive laws, using
  Proposition~\ref{prop:relations_are_kleisli} above.  See, e.g.,
  \cite{street:monads}.
For the equivalence of 1.~and 2., observe that
$T$ satisfies the BCC 
iff
$$
\xymatrix@C=15pt{
\Pre
\ar[0,1]^-{T}
&
\Pre
\ar[0,1]^-{({-})_\diamond}
&
\Rel{\Pre}
}
$$
satisfies the three properties of
Theorem~\ref{th:universal_property} above.
\qed
\end{proof}

\noindent\emphh{Remark. } 
There is an analogous theorem with ``$\Pos$'' replacing ``$\Pre$''.

\begin{corollary}\label{cor:ext-thm}
  If $T$ is a locally monotone functor, the lifting
  $\ol{T}$ is computed as 
  $$\ol{T}(R)=(Td_0)_\diamond\cdot(Td_1)^\diamond$$ where $(d_0,\E,d_1)$
  is the two-sided discrete fibration corresponding to $R$.
\end{corollary}

\medskip\noindent\emphh{Corollary. } Let $T:\Pre\to\Pre$ and
  $T_0:\Set\to\Set$ such that $TD=DT_0$ and $VT=T_0V$ where
  $V:\Pre\to\Set$ is the forgetful functor and $D$ is its
  left-adjoint. Then $T$ satisfies the BCC iff $T_0$ preserves weak
  pullbacks.

\begin{proof}
  We show that $\ol{T}$ preserves composition of relations if
  $\ol{T_0}$ does. By Corollary~\ref{cor:ext-thm} and the corollary
  after Proposition~\ref{prop:profunctors=fibrations}, we have
  $\Rel{V}\ol{T}=\ol{T_0}\Rel{V}$. Let $S,R$ be two monotone
  relations. We have $\Rel{V}\ol{T}(S\cdot R)=\ol{T_0}\Rel{V}(S\cdot
  R)=\ol{T_0}(\Rel{V}S\cdot \Rel{V}R)=\ol{T_0}(\Rel{V}S)\cdot
  \ol{T_0}(\Rel{V}R)=\Rel{V}\ol{T}(S)\cdot
  \Rel{V}\ol{T}(R)=\Rel{V}(\ol{T}S\cdot \ol{T}R)$, hence
  $\ol{T}(S\cdot R)=\ol{T}S\cdot \ol{T}R$ by $\Rel{V}$ being faithful.

  Conversely, any pullback in $\Set$ is mapped by $D$ to a pullback in
  $\Pre$ and then to an exact square by $T$. Now from
  $TD=DT_0$ and the fact that any exact square of sets is a weak
  pullback it follows that $T_0$ preserves weak pullbacks.
\end{proof}

\section{Examples}
\label{sec:examples}

\begin{example}
\label{ex:Kripke-polynomial}
All the ``Kripke-polynomial'' 
functors
satisfy the Beck-Chevalley Condition. This means the
functors defined by the following grammar:
$$
T::=
\const_\X
\mid
\Id
\mid
T^\partial
\mid
T+T
\mid
T\times T
\mid
\LL T
$$
where $\const_\X$ is the constant-at-$\X$,
$T^\partial$ is the {\em dual\/} of $T$, defined
by putting
$$
T^\partial\A=(T\A^\op)^\op
$$
and $\LL\X=[\X^\op,\Two]$ (the lowersets on $\X$, ordered
by inclusion). Observe that $\LL^\partial\X=[\X,\Two]^\op$,
hence $\LL^\partial\X=\UU\X$ (the uppersets on $\X$,
ordered by reversed inclusion).

To check that BCC is satisfied, suppose that the square
\begin{equation}
\label{eq:exact_square}
\vcenter{
\xymatrix{
\P
\ar[0,1]^-{p_1}
\ar[1,0]_{p_0}
\ar @{} [1,1]|{\nearrow}
&
\B
\ar[1,0]^{g}
\\
\A
\ar[0,1]_-{f}
&
\C
}
}
\end{equation}
is exact.
\begin{enumerate}
\item
The functor $\const_\X$.

The image of square~\refeq{eq:exact_square} under
$\const_\X$ is the square
$$
\xymatrix{
\X
\ar[0,1]^-{1_\X}
\ar[1,0]_{1_\X}
\ar @{} [1,1]|{\nearrow}
&
\X
\ar[1,0]^{1_\X}
\\
\X
\ar[0,1]_-{1_\X}
&
\X
}
$$
where the comparison is the identity. This is an exact
square (it is a Yoneda square).
\item
The functor $\Id$.

This functor obviously satisfies the Beck-Chevalley Condition.
\item
Suppose $T$ satisfies the Beck-Chevalley Condition.

The square
$$
\xymatrix{
\P^\op
\ar[0,1]^-{p_0^\op}
\ar[1,0]_{p_1^\op}
\ar @{} [1,1]|{\nearrow}
&
\A^\op
\ar[1,0]^{f^\op}
\\
\B^\op
\ar[0,1]_-{g^\op}
&
\C^\op
}
$$
is exact by Example~\ref{ex:dual}
and, by assumption, so is the square
$$
\xymatrix{
T(\P^\op)
\ar[0,1]^-{T(p_0^\op)}
\ar[1,0]_{T(p_1^\op)}
\ar @{} [1,1]|{\nearrow}
&
T\A^\op
\ar[1,0]^{T(f^\op)}
\\
T(\B^\op)
\ar[0,1]_-{T(g^\op)}
&
T(\C^\op)
}
$$
Finally, the square
$$
\xymatrix{
(T(\P^\op))^\op
\ar[0,1]^-{(T(p_1^\op))^\op}
\ar[1,0]_{(T(p_0^\op))^\op}
\ar @{} [1,1]|{\nearrow}
&
(T(\B^\op))^\op
\ar[1,0]^{(T(g^\op))^\op}
\\
(T(\A^\op))^\op
\ar[0,1]_-{(T(f^\op))^\op}
&
(T(\C^\op))^\op
}
$$
is exact by Example~\ref{ex:dual}
and this is what we were supposed to prove.
\item
Suppose both $T_1$ and $T_2$ satisfy the Beck-Chevalley Condition.
We prove that $T_1+T_2$ does satisfy it.

The image of~\refeq{eq:exact_square} under $T_1+T_2$ is
$$
\xymatrixcolsep{4pc}
\xymatrix{
T_1\P+T_2\P
\ar[0,1]^-{T_1p_1+T_2p_1}
\ar[1,0]_{T_1p_0+T_2p_0}
\ar @{} [1,1]|{\nearrow}
&
T_1\B+T_2\B
\ar[1,0]^{T_1g+T_2g}
\\
T_1\A+T_2\A
\ar[0,1]_-{T_1f+T_2f}
&
T_1\C+T_2\C
}
$$
The assertion follows from the fact that coproducts
are disjoint in $\Pre$.
\item
Suppose both $T_1$ and $T_2$ satisfy the Beck-Chevalley Condition.
We prove that $T_1\times T_2$ does satisfy it.

The image of~\refeq{eq:exact_square} under $T_1\times T_2$ is
$$
\xymatrixcolsep{4pc}
\xymatrix{
T_1\P\times T_2\P
\ar[0,1]^-{T_1p_1\times T_2p_1}
\ar[1,0]_{T_1p_0\times T_2p_0}
\ar @{} [1,1]|{\nearrow}
&
T_1\B\times T_2\B
\ar[1,0]^{T_1g\times T_2g}
\\
T_1\A\times T_2\A
\ar[0,1]_-{T_1f\times T_2f}
&
T_1\C\times T_2\C
}
$$
The assertion follows from how products
are formed in $\Pre$.
\item
Suppose that $T$ satisfies the Beck-Chevalley Condition.
We prove that $\LL T$ does satisfy it again.

It suffices to prove that $\LL$ satisfies the Beck-Chevalley
Condition. The image of square~\refeq{eq:exact_square} under
$\LL$ is the square
$$
\xymatrix{
\LL\P
\ar[0,1]^-{\LL p_1}
\ar[1,0]_{\LL p_0}
\ar @{} [1,1]|{\nearrow}
&
\LL\B
\ar[1,0]^{\LL g}
\\
\LL\A
\ar[0,1]_-{\LL f}
&
\LL\C
}
$$
First recall how $\LL$ is defined on monotone maps:
for example, $\LL f:\LL\A\to\LL\C$ is defined as
a left Kan extension along $f^\op:\A^\op\to\C^\op$.
This means that, for every lowerset $W:\A^\op\to\Two$,
$$
(\LL f)(W)
=
\bigvee_a \C^\op(f^\op a,{-})\wedge Wa
$$
or, in a more readable fashion,
$$
(\LL f)(W):
c\mapsto
\bigvee_a \C(c,fa)\wedge Wa
$$
Hence $c$ is in the lowerset $(\LL f)(W)$
iff there is $a$ in $W$ such that $c\leq fa$.
Observe that $\LL$ is indeed
a functor: it clearly preserves identities and
composition (for that, see Theorem~4.47 of~\cite{kelly:book})
up to isomorphisms. But these canonical isomorphisms are
identities, since $[\X^\op,\Two]$ is always a {\em poset\/}.

We employ Example~\ref{ex:square+adjoints}:
both $\LL f$ and $\LL p_1$ are left adjoints with
$(\LL f)^r=[f^\op,\Two]$ and $(\LL p_1)^r=[p_1^\op,\Two]$.
Hence it suffices to prove that 
$$
\LL p_0\o [p_1^\op,\Two]
=
[f^\op,\Two]\o \LL g
$$
Moreover, by the density of principal lowersets
of the form $\B({-},b_0)$ in $\LL\B$ and the fact that
all the monotone maps $\LL p_0$, $[p_1^\op,\Two]$,
$[f^\op,\Two]$, $\LL g$ preserve suprema
(since they all are left adjoints), it suffices
to prove that
\begin{equation}
\label{eq:L}
(\LL p_0\o [p_1^\op,\Two])(\B({-},b_0))
=
([f^\op,\Two]\o \LL g)(\B({-},b_0))
\end{equation}
holds for all $b_0$.

The left-hand side is isomorphic to
$$
\LL p_0 (\B(p_1{-},b_0))
=
a\mapsto \bigvee_w \A(a,p_0 w)\wedge \B(p_1 w,b_0)
$$
By exactness of~\refeq{eq:exact_square}, this means that
$$
\LL p_0 (\B(p_1{-},b_0))
=
a\mapsto\C(fa,gb_0)
$$
Observe further that
$$
(\LL g)(\B({-},b_0))
=
c\mapsto
\bigvee_b \C(c,gb)\wedge \B(b,b_0)
$$
hence
$$
(\LL g)(\B({-},b_0))
=
c\mapsto
\C(c,gb_0)
$$
by the Yoneda Lemma.

The right hand side of~\refeq{eq:L} is therefore
isomorphic to
$$
([f^\op,\Two]\o \LL g)(\B({-},b_0))
=
[f^\op,\Two](c\mapsto \C(c,gb_0))
=
a\mapsto\C(fa,gb_0)
$$
\end{enumerate}
\end{example}

\begin{example}
\label{ex:pre-to-pos}Recall the adjunction $Q\dashv I:\Pos\to\Pre$,
where $I$ is the inclusion functor and $Q(\A)$ is the quotient of $\A$
obtained by identifying $a$ and $b$ whenever $a\le b$ and $b\le
a$. The functors $Q$ and $I$ are locally monotone and map exact
squares to exact squares. Hence, if $T:\Pre\to\Pre$ satisfies the BCC,
so does $QTI:\Pos\to\Pos$.

\end{example}

\begin{example}\label{ex:powerset}The {\em powerset functor\/} $\PP:\Pre\to \Pre$ is
defined as follows. The order on $\PP\A$ is the Egli-Milner preorder, that is, $\PP(A,B)=1$
if and only if
\begin{equation}
\forall a\in A\ \exists b\in B\ a\le b
\textrm{ and }
\forall b\in B \ \exists a\in A\ a\le b
\end{equation}
$\PP f(A)$ is the direct image of $A$.
The functor $\PP$ is locally monotone and satisfies the
BCC.

The {\em finitary powerset functor} $\PP_\omega$ is defined similarly: $\PP_\omega\A$ consists of the finite subsets of $\A$ equipped with the Egli-Milner preorder.   $\PP_\omega$ is locally monotone and satisfies the
BCC.

The powerset functor $\mathbb P$ is locally monotone and satisfies the
BCC. This follows from the unnumbered corollary of
Section~\ref{sec:extension}. For a direct argument consider an exact
square:
\begin{equation}
\vcenter{
\xymatrix{
\P
\ar[0,1]^-{p_1}
\ar[1,0]_{p_0}
\ar @{}[1,1]|{\nearrow}
&
\B
\ar[1,0]^{g}
\\
\A
\ar[0,1]_-{f}
&
\C
}
}
\end{equation}
By~\refeq{eq:exactness_coend} we have to show that for
$A\in\PP\A$ and $B\in\PP\B$
\begin{equation}
\PP\C(\PP f(A),\PP g(B))
=
\bigvee_W \PP\A(A,\PP p_0 (W))\wedge\PP\B(\PP p_1(W),B)
\end{equation}
Assume $\PP\C(\PP f(A),\PP g(B))=1$. Then
\begin{equation}
\label{equ:conv-pres-ex-sq}
\forall a\in A\ \exists b\in B\ fa\le gb
\textrm{ and }
\forall b\in B \ \exists a\in A\ fa\le gb
\end{equation}
We have to find $W\in\PP\P$ such that
$\PP\A(A,\PP p_0 (W))$ and $\PP\B(\PP p_1(W),B)$.
Let
$
W=
\{
w\in\P\mid\exists a\in A\ a\le p_0w \textrm{ and } \exists b\in B\ p_1w\le b
\}
$.
It is easy to see that $W$ satisfies
$\forall w\in W\ \exists a\in A\ \A(a,p_0w)$ and
$\forall w\in W\ \exists b\in B\ \B(p_1w,b)$. Consider
$a\in A$. By~\refeq{equ:conv-pres-ex-sq} there exists $b\in B$
such that $\C(fa,gb)$. By~\refeq{eq:exactness_coend} there exists
$w\in W$ such that $\A(a, p_0w)$. So
$\PP\A(A,\PP p_0 (W))=1$.
Similarly, we can show that for all $b\in B$ exists
$w\in W$ with $\B(p_1w,b)$. This shows that $\PP$ preserves exact
squares, hence it satisfies the BCC.

The proof that
$\PP_\omega$ satisfies the BCC goes along the same lines.
\end{example}

\begin{example}
\label{ex:convex}
Given a preorder $\A$, a subset $A\subseteq \A$ is called
{\em convex\/} if $x\le y\le z$ and $x,z\in A$ imply $y\in A$.

The {\em convex powerset functor\/} $\PP^c:\Pos\to \Pos$ is defined as
follows. $\PP^c\A$ is the set of convex subsets of $\A$ endowed with
the Egli-Milner order. $\PP^c f(A)$ is the direct image of $A$. This
is a well defined locally monotone functor. Notice that $\PP^c\cong
Q\PP I$. This follows from the fact that if $\A$ is a poset and
$A,B\in\PP I\A$, then $\PP I\A(A,B)=1$ and $\PP I\A(B,A)=1$ if and
only if $A$ and $B$ have the same convex hull.  Hence, by
Example~\ref{ex:pre-to-pos}, $\PP^c$ satisfies the BCC.

The {\em finitely-generated convex powerset\/} $\PP^c_\omega$
is defined similarly to $\PP^c$. The only difference is that
the convex sets appearing in $\PP^c_\omega\A$ are convex hulls
of finitely many elements of $\A$. Then  $\PP^c_\omega$
is locally monotone and is isomorphic to $Q\PP_\omega I$, thus it also
satisfies the BCC. Again, we have 
that $\PP^c_\omega=Q\PP_\omega I$ and $\PP^c_\omega$ satisfies the BCC.

Observe that both functors are self-dual: $(\PP^c)^\partial=\PP^c$
and $(\PP^c_\omega)^\partial=\PP^c_\omega$.
\end{example}

\begin{example}
Since the lowerset functor $\LL:\Pre\to\Pre$
satisfies the Beck-Chevalley Condition by
Example~\ref{ex:Kripke-polynomial},
we can compute its lifting
$\ol{\LL}:\Rel{\Pre}\to\Rel{\Pre}$.
We show how $\ol{\LL}$ works on the
relation
$
\xymatrix@1{
\A
\ar[0,1]|-{\object @{/}}^-{R}
&
\B
}
$.
The value $\ol{\LL}(R)$
is, by Theorems~\ref{th:universal_property}
and~\ref{th:extension}, given by
$(\LL d_0)_\diamond \o (\LL d_1)^\diamond$
where
$(d_0,\E^R,d_1):\A\to \B$
is the two-sided discrete fibration corresponding to $R$.
Using the formula~\refeq{eq:relation_composition}
for relation composition, we can write
\begin{equation}\label{equ:acoend}
\ol{\LL}(R)(B,A)
=
\bigvee_W \LL\B(B,\LL d_0 (W))\wedge\LL\A(\LL d_1(W),A)
\end{equation}
where $B:\B^\op\to\Two$ and $A:\A^\op\to\Two$ are arbitrary
lowersets.
Since $\LL d_1$ is a left adjoint to restriction along
$d_1^\op:(\E^R)^\op\to\A^\op$, we can rewrite~\refeq{equ:acoend}
to
$$
\ol{\LL}(R)(B,A)
=
\bigvee_W \LL\B(B,\LL d_0 (W))\wedge\LL\E^R(W,A\o d_1^\op)
$$
and, by the Yoneda Lemma, to
$$
\ol{\LL}(R)(B,A)
=
\LL\B(B,\LL d_0 (A\o d_1^\op))
$$
Hence the lowersets $B$ and $A$ are related by $\ol{\LL}(R)$
if and only if the inclusion
$$
B\subseteq \LL d_0 (A\o d_1^\op)
$$
holds in $[\B^\op,\Two]$.
Recall that
$$
\LL d_0 (A\o d_1^\op)
=
b\mapsto\bigvee_w \B(b,d_0 w)\wedge (A\o d_1^\op)(w)
$$
Therefore the inclusion $B\subseteq \LL d_0 (A\o d_1^\op)$
is equivalent to the statement:
For all $b$ in $B$ there is $(b_1,a_1)$
such that $R(b_1,a_1)$ and $b\leq b_1$ and $a_1$ in $A$.

Observe that the above condition is reminiscent of
one half of the Egli-Milner-style of the relation
lifting of a powerset functor. This is because
$\LL$ is the ``lower half'' of two possible
``powerpreorder functors''. The ``upper half''
is given by $\UU:\Pre\to\Pre$ where $\UU=\LL^\partial$.
\end{example}

\begin{example}
\label{ex:convex_lifted}
The relation liftings $\ol{\PP}$, $\ol{\PP^c}$, $\ol{\PP_\omega}$,
$\ol{\PP^c_\omega}$ of the (convex) powerset functor and their
finitary versions yield the ``Egli-Milner'' style of the relation
lifting. More precisely, for a relation $ \xymatrix@1{ \B
  \ar[0,1]|-{\object@{/}}^-{R} & \A } $ we have $\ol{\PP}(R)(B,A)$
(respectively $\ol{\PP_{\omega}}(R)(B,A)$, $\ol{\PP^{c}}(R)(B,A)$,
$\ol{\PP^{c}_{\omega}}(R)(B,A)$) if and only if
$$\forall a\in A\ \exists b\in B\ R(b,a)\textrm{ and }
\forall b\in B\ \exists a\in A\ R(b,a).$$
To compute the lifting of $\PP^c$, consider a
monotone relation
$
\xymatrix@1{
\A
\ar[0,1]|-{\object@{/}}^-{R}
&
\B
}
$
and the induced fibration $(d_0,\E,d_1):\A\to\B$. We know that
$\ol{\PP^c}(R)=(\PP^c d_0)_\diamond \o (\PP^c d_1)^\diamond$, so
\begin{equation}
\ol{\PP^c}(R)(B,A)
=
\bigvee_E \PP^c\B(B,\PP^c d_0(E))\wedge\PP\A(\PP^c d_1(E),A)
\end{equation}

We prove that $\ol{\PP^c}(R)(B,A)=1$ implies
$\forall a\in A\ \exists b\in B\ R(b,a)$ and
$\forall b\in B\ \exists a\in A\ R(b,a)$.
Consider a witness $E$ and $a\in A$. Since
$\PP^c\A(\PP^c d_1 (E),A)=1$, there exists $(b',a')\in E$ such that
$\A(a',a)$. Since $\PP^c\B(B,\PP^c d_0(E))=1$, there exists
$b\in B$ such that $\B(b,b')$. Since $R$ is monotone and
$R(b',a')=1$ we obtain $R(b,a)=1$. So
$\forall a\in A\ \exists b\in B\ R(b,a)$.
The second part is analogous.

Conversely, if
$\forall a\in A\ \exists b\in B\ R(b,a)$ and
$\forall b\in B\ \exists a\in A\ R(b,a)$, define
the subset of $\E$ as follows:
$$
E=
\{
\xymatrix@1{b\ar@{~>}[0,1]&a}
\mid
b\in B,\
a\in A
\}
$$
Then $E$ is convex, since both $B$ and $A$ are convex. Both $\PP^c\B(B,\PP^c d_0 (E))=1$ and $\PP^c\A(\PP^c d_1(E),A)=1$ hold for obvious
reasons. Hence $\ol{\PP^c}(R)(B,A)=1$ holds.
\end{example}

\begin{example}
\label{ex:notBCC}
To find a functor that does not satisfy the BCC,
it suffices, by Remark~\ref{rmk:BCC}, to find a locally monotone functor
$T:\Pre\to\Pre$ that does not preserve
order-embeddings.
For this, let $T$ be the {\em connected components functor\/},
i.e., $T$ takes a preorder $\A$ to the discretely ordered
poset of connected components of $\A$.
$T$ does not preserve embedding $f:\A\to\B$ indicated below.
$$
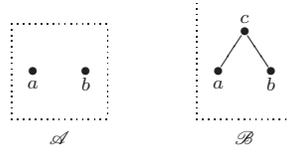
\let\objectstyle=\scriptstyle
\xy <1 pt,0 pt>:
    (000,000)  *++={};
    (030,030) *++={} **\frm{.};
    (070,000)  *++={};
    (100,040) *++={} **\frm{.}
\POS(015,-10) *{\A} = "A";
    (085,-10) *{\B} = "B";
    (005,015) *{\bullet};
    (025,015) *{\bullet};
    (005,010) *{a} = "aA";
    (025,010) *{b} = "bA";
    (075,015) *{\bullet} = "a0B";
    (095,015) *{\bullet} = "b0B";
    (075,010) *{a} = "aB";
    (095,010) *{b} = "bB";
    (085,030) *{\bullet} = "c0B";
    (085,035) *{c} = "cB";
\POS "a0B" \ar@{-} "c0B";
\POS "b0B" \ar@{-} "c0B";
\endxy
$$
\end{example}

\section{An Application: Moss's Coalgebraic Logic over Posets}
\label{sec:Coalg-posets} We show how to develop the basics of
Moss's coalgebraic logic over posets. For reasons of space, this
development will be terse and assume some familiarity with, e.g.,
Sections~2.2 and~3.1 of~\cite{kurz-leal:mfps09}.

Since the logics will have propositional connectives but no
negation (to capture the semantic order on the logical side)
we will use the category $\DL$ of bounded distributive
lattices. We write $F\dashv U:\DL\to\Pos$ for the obvious adjunction;
and $P:\Pos^\op\to\DL$ where $UP\X=[\X,\Two]$ and $S:\DL\to\Pos^\op$
where $SA=\DL(A,\Two)$.
Note that $UP=[-,\Two]$ and recall $\LL=[({-})^\op,\Two]$.
Further, let $T:\Pos\to\Pos$ be a locally monotone finitary functor
that satisfies the BCC.

We define coalgebraic logic abstractly by a functor
$L:\DL\to\DL$ given as
$$
L=FT^\partial U
$$
where the functor $T^\partial:\Pos\to\Pos$ is given by $ T^\partial \X
= (T(\X^\op))^\op$.
By Example~\ref{ex:Kripke-polynomial}, $T^\partial$ satisfies the BCC.
The formulas of the logic are the elements of the initial $L$-algebra
$FT^\partial U (\Lang)\to \Lang$.  The formula given by some
$\alpha\in T^\partial U (\Lang)$ is written as
$\nabla\alpha.$
The semantics is given by a natural transformation
$$
\delta:LP\to PT^\op
$$
Before we define $\delta$, we need for every preorder $\A$, the
relation\footnote{The type of $\ni_{\X}$ conforms with the logical reading
  of $\ni$ as $\Vdash$. Indeed, ${\ni}(x,\phi)\ \& \ \phi\subseteq\psi\
  \Rightarrow\ {\ni}(x,\psi)$ and ${\ni}(x,\phi) \ \&\ x\le y\
  \Rightarrow\ {\ni}(y,\phi)$, where $\phi,\psi$ are uppersets of $\X$.
}
$$
\xymatrix@1{
[\A,\Two]
\ar[0,1]|-{\object @{/}}^-{\ni_\A}
&
\A^\op
}
$$
given 
by the evaluation map $\ev_\A:\A\times [\A,\Two]\to\Two$. Observe that
\begin{equation}\label{equ:ni-yon}
\ni_\A=(\yon_{\A^\op})^\diamond
\end{equation}
since $ (\yon_{\A^\op})^\diamond(a,V)= [\A,\Two](\yon_{\A^\op}a,V)= Va
$ holds by the Yoneda Lemma.

\begin{lemma}\label{lem:ni}
For every monotone map $f:\A\to\B$ we have
$$
\xymatrix@R=16pt{
[\A,\Two]
\ar[0,2]|-{\object @{/}}^-{\ni_\A}
&
&
\A^\op
\\
[\B,\Two]
\ar[0,2]|-{\object @{/}}_-{\ni_\B}
\ar[-1,0]|-{\object @{/}}^-{[f,\Two]^\diamond\ }
&
&
\B^\op
\ar[-1,0]|-{\object @{/}}_-{\ (f^\op)^\diamond}
}
$$
\end{lemma}

\begin{proof}[Diagrammatic Proof.]
The square
$$
\xymatrix{
\A^\op
\ar[0,2]^-{\yon_{\A^\op}}
\ar[1,0]_{f^\op}
&
&
\LL(\A)
\ar[1,0]^{\LL (f)}
\\
\B^\op
\ar[0,2]_-{\yon_{\B^\op}}
&
&
\LL(\B)
}
$$
commutes in $\Pre$, since $\yon$ is natural.
Hence the square
$$
\xymatrix{
\A^\op
\ar[0,2]|-{\object @{/}}^-{(\yon_{\A^\op})_\diamond}
\ar[1,0]|-{\object @{/}}_{(f^\op)_\diamond}
&
&
\LL(\A)
\ar[1,0]|-{\object @{/}}^{(\LL (f))_\diamond}
\\
\B^\op
\ar[0,2]|-{\object @{/}}_-{(\yon_{\B^\op})_\diamond}
&
&
\LL(\B)
}
$$
commutes in $\Rel{\Pre}$ since $({-})_\diamond$ is a 2-functor.

Now observe that $\LL(f)\dashv [f,\Two]$ holds by the definition
of $\LL$ on morphisms. Hence $(\LL(f))^\diamond\dashv [f,\Two]^\diamond$
holds. Since adjoints are determined uniquely up to isomorphisms,
this shows that $(\LL(f))_\diamond=[f,\Two]^\diamond$ (we use that
isomorphisms are identities in $\Rel{\Pre}$).

Thus, taking right adjoints everywhere in the above square
we obtain the square from the claim of the lemma.
\end{proof}

\begin{proof}[Computational Proof.]
By definition
\begin{eqnarray*}
{\ni_\A}\o [f,\Two]^\diamond(a,V)
&=
&
\bigvee_W {\ni_\A}(a,W)\wedge [\A,\Two](V\o f,W)
\\
&=
&
{\ni_\A}(a,V\o f)
\\
&=
&
(V\o f)(a)
\end{eqnarray*}
where the second 
step is due to the Yoneda Lemma.
Analogously:
\begin{eqnarray*}
(f^\op)^\diamond\o{\ni_\B}(a,V)
&=
&
\bigvee_b \B^\op(f^\op a,b)\wedge {\ni_\B}(b,V)
\\
&=
&
{\ni_\B}(fa,V)
\\
&=&
V(fa)
\end{eqnarray*}
\end{proof}

\begin{corollary}
\label{cor:tau}
For every locally monotone functor $T$ that
satisfies the Beck-Chevalley Condition and for
every monotone map $f:\A\to\B$, we have 
$$
\xymatrix@R=16pt{
\ol{T}[\A,\Two]
\ar[0,2]|-{\object @{/}}^-{\ol{T}{\ni_\A}}
&
&
\ol{T}\A^\op
\\
\ol{T}[\B,\Two]
\ar[0,2]|-{\object @{/}}_-{\ol{T}{\ni_\B}}
\ar[-1,0]|-{\object @{/}}^-{\ol{T}[f,\Two]^\diamond\ }
&
&
\ol{T}\B^\op
\ar[-1,0]|-{\object @{/}}_-{\ \ol{T}(f^\op)^\diamond}
}
$$
\end{corollary}

\noindent Coming back to $\delta:LP\to PT^\op$. It suffices, due to
$F\dashv U$, to give 
$$
\tau:T^\partial UP\to UPT^\op
$$
Observe that, for every preorder $\X$, we have
$$
UPT^\op(\X)=
[T^\op\X,\Two]=
\LL ((T^\op\X)^\op)
$$
By Proposition~\ref{prop:relations_are_kleisli}, to define $\tau_\X$
it suffices to give a relation from
$T^\partial UP\X$ to $(T^\op\X)^\op$, and we obtain it from
Theorem~\ref{th:extension} by applying $\ol{T^\partial}$ to the
relation $\ni_\X$. That $\tau_\X$ so defined is natural, follows
from Corollary~\ref{cor:tau}. This follows~\cite{kkv:aiml08}
with the exception that here now we need to use $T^\partial$.

\begin{example}
Recall the functor $\PP^c_\omega$ of Example~\ref{ex:convex}
and consider a
coalgebra $c:\X\to \PP^c_\omega\X$. On the logical side we allow ourselves to write
$\nabla\alpha$ for any finite subset $\alpha$ of $U (\Lang)$. Of
course, we then have to be careful that the
semantics of $\alpha$ agrees with the semantics of the convex
closure of $\alpha$. Interestingly, this is done automatically by
the machinery set up in the previous section, since $\PP^c_\omega=Q\PP_\omega I$ and all these functors are self-dual. By Example~\ref{ex:convex_lifted}, the semantics of $\nabla\alpha$ is given by
$$
x\Vdash\nabla\alpha
\quad
\Leftrightarrow
\quad
\forall y\in c(x)\exists\phi\in\alpha.y\Vdash\phi
\ \textrm{ and }
\forall\phi\in \alpha\exists y\in c(x).y\Vdash\phi.
$$
\end{example}

\section{Conclusions}
\label{sec:conclusions}
We hope to have illustrated in the previous two sections that, after
getting used to handle the $(-)_\diamond, (-)^\diamond$ and $(-)^\op$,
the techniques developed here work surprisingly smoothly and will be
useful in many future developments. For example, an observation
crucial for both \cite{kkv:aiml08,kurz-leal:mfps09} is that composing
the singleton map $X\to\P X$, $x\mapsto\{x\}$, with the relation
$\xymatrix@1{\ni_X:\P X\ar[0,1]|-{\object@{/}}& X}$ is $\id_X$.
Referring back to~\refeq{equ:ni-yon}, we find here the same
relationship
$$
{\ni}_\A\circ (\yon_{\A^\op})_\diamond
=
(\yon_{\A^\op})^\diamond\circ(\yon_{\A^\op})_\diamond=\id_{\A^\op}
$$
The question whether the completeness proof of \cite{kkv:aiml08} and
the relationship between $\nabla$ and predicate liftings of
\cite{kurz-leal:mfps09} can be carried over to our setting are a
direction of future research.

Another direction is the generalisation to categories which
are enriched over more general structures than $\Two$, such as
commutative quantales. Simulation, relation lifting and final
coalgebras in this setting have been studied in \cite{worrell:cmcs00}.

\end{document}